\documentclass[a4paper,12pt,reqno]{amsart}

\usepackage{hyperref}
\usepackage{amsmath,amssymb}
\usepackage{graphicx}
\usepackage{multicol}
\usepackage{enumerate}

\newtheorem{theorem}{Theorem}
\newtheorem{lemma}{Lemma}[section]
\newtheorem{proposition}[lemma]{Proposition}

\newtheorem{remark}[lemma]{Remark}

\numberwithin{equation}{section}

\newcommand{\pd}[2]{\frac{\partial {#1}}{\partial {#2}}}
\newcommand{\CO}{\mathbb C}

\newcommand{\RE}{\mathbb R}

\newcommand{\DD}{\mathcal D}
\newcommand{\BB}{\mathcal B}
\newcommand{\EE}{\mathcal E}

\newcommand{\FF}{\mathcal F}
\newcommand{\ve}{\varepsilon}
\newcommand{\ga}{\gamma}

\newcommand{\de}{\delta}

\newcommand{\ome}{\omega}

\newcommand{\xx}{{\bf x}}
\newcommand{\kk}{{\bf k}}
\newcommand{\oo}{{\bf 0}}
\newcommand{\dk}{\textnormal{d}\kk\,}
\newcommand{\dx}{\textnormal{d}\xx\,}
\newcommand{\ds}{\textnormal{d}s\,}
\newcommand{\dt}{\textnormal{d}t\,}

\newcommand{\dd}{\textnormal{d}}

\newcommand{\n}{\noindent}
\newcommand{\lf}{\left}
\newcommand{\ri}{\right}

\renewcommand{\Re}{\operatorname{Re}\,}
\renewcommand{\leqslant}{\leq}

\newcommand{\f}{\frac}
\newcommand{\ba}{\begin{eqnarray}} \newcommand{\ea}{\end{eqnarray}}
\newcommand{\bdm}{\begin{displaymath}} \newcommand{\edm}{\end{displaymath}} 
\newcommand{\brr}{\begin{array}}\newcommand{\err}{\end{array}} 

\title[]{The point-like limit for a NLS equation with concentrated nonlinearity in dimension three}
\author[]{Claudio Cacciapuoti}
\address[Claudio Cacciapuoti]{Dipartimento di Scienza e Alta Tecnologia, Universit\`a dell'Insubria, Via Valleggio 11, 22100 Como, Italy}
\email{claudio.cacciapuoti@uninsubria.it}%

\author[]{Domenico Finco}
\address[Domenico Finco]{Facolt\`a di Ingegneria, Universit\`a Telematica
Internazionale Uninettuno,  Corso Vittorio Emanuele II 39, 00186 Roma,
Italy}
\email{d.finco@uninettunouniversity.net}

\author[]{Diego Noja}
\address[Diego Noja]{Dipartimento di Matematica e Applicazioni, Universit\`a
 di Milano Bicocca,  via R. Cozzi 55, 20125 Milano, Italy}
\email{diego.noja@unimib.it} 

\author[]{Alessandro Teta}
\address[Alessandro Teta]{Dipartimento di Matematica G. Castelnuovo, Sapienza Universit\`a di Roma,  Piazzale  Aldo Moro 5, 00185 Roma, Italy}
\email{teta@mat.uniroma1.it}

\date{}

\thanks{
D.F. and D.N.  acknowledge the support of FIRB 2012 project ``Dispersive dynamics: Fourier Analysis and Variational Methods'', Ministry of University and
Research of Italian Republic  (code RBFR12MXPO).
C.C. acknowledges the support of the FIR 2013 project ``Condensed Matter in Mathematical Physics'', Ministry of University and
Research of Italian Republic  (code RBFR13WAET)}

\begin{document}
\begin{abstract}
We consider a scaling limit of a nonlinear Schr\"odinger equation (NLS) with a nonlocal nonlinearity showing that it reproduces in the limit of cutoff removal a NLS equation with nonlinearity concentrated at a point. The regularized dynamics is described by the equation
\begin{equation*}
i\pd{}{t}\psi^\ve(t)= -\Delta \psi^\ve(t) + g(\varepsilon,\mu,|(\rho^\ve,\psi^\ve(t))|^{2\mu}) (\rho^\ve,\psi^\ve(t)) \rho^\ve\ 
\end{equation*}
where $\rho^{\ve} \to \delta_0$ weakly and the function $g$ embodies the nonlinearity and the scaling and has to be fine tuned in order to have a nontrivial limit dynamics.
The limit dynamics is a nonlinear version of point interaction in dimension three and it has been previously studied in several papers as regards the well-posedness, blow-up and asymptotic properties of solutions. Our result is the first justification of the model as the point limit of a regularized dynamics. 
\end{abstract}

\maketitle

\begin{footnotesize}
 
 \end{footnotesize}

\vspace{1cm}

\section{Introduction}

In this paper we obtain a nonlinear Schr\"odinger dynamics with a nonlinearity concentrated at a point in dimension three as a scaling limit of a regularized nonlinear Schr\"odinger dynamics.\\
We consider the abstract nonlinear Schr\"odinger equation 
\begin{equation}
\label{limiteq}
i\pd{}{t}\psi(t)=H_{\mu, \gamma}\psi(t)\ ,\ \ \ \ \ \psi(0)=\psi_0\ \in \DD 
\end{equation}
where $H_{\mu, \gamma}$, with $\gamma\in\RE$ and $ \mu \geq 0$, is the nonlinear operator with domain given by the nonlinear manifold 
\begin{equation*}
\DD (H_{\mu,\gamma}):=\DD=  \bigg\{ \psi\in L^2 (\RE^3)| \, \psi= \phi +q G;\,\phi \in \dot H^2 (\RE^3) ,\, q\in \CO;\,    \phi(\oo) = \gamma |q|^{2\mu} q \bigg\},
\end{equation*}
where 
\begin{equation*}
G(x)= \frac{1}{4\pi|x|},
\end{equation*}
and with action given by
\begin{equation*}
H_{\mu, \gamma}\psi=-\Delta \phi\ .
\end{equation*}
Eq.  \eqref{limiteq} characterizes what we call the {\it limit problem}, and its solution is the {\it limit flow}. It describes a nonlinear Schr\"odinger dynamics with a nonlinearity concentrated in a single point, fixed for simplicity at the origin. The nonlinearity is attractive (or focusing) if $\gamma<0$, is repulsive (or defocusing) if $\gamma>0$. This is a non standard evolution problem, because an element of the nonlinear domain is the sum of a \emph{regular part} $\phi\in \dot H^2(\RE^3)$ and a \emph{singular part} $qG$ which is in $L_{loc}^2(\RE^3)$ but not in $\dot H^2(\RE^3)$ (not even in $\dot H^1(\RE^3)$). 
The denomination of concentrated nonlinearity is justified as follows. 
First, if we consider a $\psi \in \mathcal D$ which vanishes in a neighborhood of the origin one has $H_{\mu, \gamma}\psi=-\Delta \psi$, i.e., the interaction is supported at the origin. Furthermore, at the origin a nonlinear boundary condition embodied in the definition of the operator domain is given, namely  $ \phi(\oo) = \gamma |q|^{2\mu} q$. This boundary condition relates in a nonlinear way the regular part of the field with the coefficient $q$ of the singular part, usually called charge.

\n
In the physical literature such a boundary condition is often written in the equivalent form
\begin{equation}\label{scle}
\frac{\partial}{\partial r}(r\psi)(0)=\alpha(r\psi)(0).
\end{equation}
In the linear case ($\mu=0$), the one which is  mostly treated in the physical literature, $\alpha$ is a constant related to the {\it scattering length}; the nonlinear model is formally obtained by setting   
$\alpha=\gamma(4\pi)^{2\mu+1}|(r\psi(0))|^{2\mu}$.

\n
About the linear model, also called contact or point interaction or more informally $\delta$ pseudo-potential, a huge literature exists, and we refer to \cite{Albeverio} as a comprehensive treatment. Concerning the nonlinear model, the well-posedness of the Cauchy problem for \eqref{limiteq} was studied in its weak form (energy domain), in \cite{at2}, establishing also usual NLS conservation laws of mass ($L^2$-norm) and energy. The blow-up of solutions of \eqref{limiteq} was studied instead in \cite{at3}.  More recently asymptotic stability of standing waves of the model has been treated in \cite{Cecilia1}. A wave dynamics with the same nonlinear generator appears in \cite{NP}, and on the same lines in \cite{Kopylova16} the Klein Gordon Cauchy problem with concentrated nonlinearity is considered (see also \cite{Kopylova17}). Much better known is the one dimensional model (see \cite{at1,CFNT1,Claudio,HolmerLiu15} and references therein). Only recently the two dimensional case has been considered, see \cite{CCT, CFTjfa17}.\\
To introduce the main motivation of our work, we stress the fact that the definition of nonlinear operator $H_{\mu,\gamma}$ has been given in \cite{at2} only by analogy with the linear case. So the question arises to find a satisfactory justification of such a model. This is the problem  studied here. Some remarks are in order. While the natural introduction of the linear model is through the theory of self-adjoint extensions (see e.g. \cite{Albeverio}), in the nonlinear setting there is no such a possibility and the only conceivable strategy seems to resort to an approximation procedure from a regular model. 
There are at least two different and both natural ways to represent a nonlinearity with a small support around zero. The first one is the {\it local} inhomogeneous nonlinearity defined by
\begin{equation}\label{local}
i\pd{}{t}\psi^\ve(t)= -\Delta \psi^\ve(t) + V_{\varepsilon}(x)|\psi^\varepsilon(t)|^{2\mu}\psi^\varepsilon(t)\  ,\ \ \mu>0
\end{equation}
with $V_{\varepsilon}=s(\varepsilon)V(\frac{x}{\varepsilon}),$ $V$ integrable and $s(\varepsilon)$ a suitable scaling function. Convergence of the flow \eqref{local} in {\it dimension one}   has been studied  in \cite{CFNT1}, and in this case the correct scaling is the $\delta$ scaling $s(\varepsilon)=\frac{1}{\varepsilon}$. The limit dynamics in this one dimensional case is related to a ``nonlinear delta potential'', a subject which appears also in a number of physical applications (see references in \cite{CFNT1}). 
A second procedure is a nonlocal or ``mean field'' approximation which consists in smearing the wavefunction $\psi$ appearing in the nonlinear term by means of a form factor $\rho^{\varepsilon}$, the support of which is shrinking to zero when $\varepsilon $ is going to zero:
\begin{equation*}
i\pd{}{t}\psi^\ve(t)= -\Delta \psi^\ve(t) + g(\varepsilon,\mu,|(\rho^\ve,\psi^\ve(t))|^{2\mu}) (\rho^\ve,\psi^\ve(t)) \rho^\ve\ .
\end{equation*}
Here the function $g(\ve,\mu,|(\rho^\ve,\psi^\ve(t))|^{2\mu})$ depends on the form factor $\rho^{\ve}$, on the power $\mu$ and contains the self-coupling term. It has to be gauged in such a way that a limit dynamics exists when the regularization is removed, i.e., when $\rho^\ve \stackrel{w}\longrightarrow\delta_0$.\\ In the present paper we consider the three dimensional case following this second procedure. Specifically, we choose 
\[
\rho^\ve(x)=\frac{1}{\ve^3}\, \rho \! \left(\frac{x}{\ve} \right)
\]
where  $\rho$ belongs to  the Schwartz space $ \mathcal S(\RE^3)$, is real, positive, spherically symmetric and satisfies $\int \!\rho(x) dx=1$ (actually, much less regularity is needed for the proof).
 The mean field regularization and its limit, in the much simpler one dimensional case, has been treated in \cite{KK07}.\\
The crucial choice of the form of $g$ is guided by an analysis of the corresponding linear problem. In this case it is known that the only possibility is given by
\begin{equation*}
g(\ve)=-\frac{\ve}{\ell}+ \gamma \, \frac{\ve^2}{\ell^2}+\text{o}(\ve^2)
\end{equation*}
where
\begin{equation}\label{M}
\ell=(\rho,(-\Delta)^{-1}\rho) 
\end{equation}
and $\gamma$ is an arbitrary real constant. This result (see \cite{Albeverio}) comes from the analysis of the norm resolvent convergence of the right hand side of \eqref{local}, which of course implies the convergence of the evolution groups.
The two terms in the previous formula have different roles.  The term $-\ve/\ell$ is needed to guarantee a cancellation of a divergent term arising from $-\Delta\psi^\ve(t)$; moreover it  shows that in the limit the quantity $\ve(\rho^\ve,\psi^\ve(t))$ is convergent. The term $\gamma \ve^2/\ell^2$  gives in the limit  the scattering length of the interaction defined in Eq. \eqref{scle} (see also for a similar analysis in a different context \cite{NP2}).
\\
These considerations suggest a nonlinear coupling of the form
\begin{equation*}
g(\varepsilon,\mu,|(\rho^\ve,\psi^\ve(t))|^{2\mu})= - \frac{\ve}{\ell} + \gamma \, \frac{\ve^2}{\ell^2} \left| \frac{\ve}{\ell}(\rho^\ve,\psi^\ve(t)) \right|^{2\mu}.
\end{equation*}
With these premises, we introduce the evolution problem (from now on the {\it scaled problem})
\begin{equation}
\label{appeq}
i\pd{}{t}\psi^\ve(t)= -\Delta \psi^\ve(t) + \left( - \frac{\ve}{\ell} + \gamma\frac{\ve^{2+2\mu}}{\ell^{2+2\mu}} |(\rho^\ve,\psi^\ve(t))|^{2\mu} \right) (\rho^\ve,\psi^\ve(t)) \rho^\ve
\end{equation} with initial datum 
\begin{equation}\label{appid}
\psi^\ve(0) = \psi^\ve_0
\end{equation}
where   $\gamma\in\RE$ and $ \mu \geq 0 $.\\
We want to show that the solutions of the scaled problem converge to the solutions of the limit problem for a reasonably large class of initial data. 
Preliminarily, we notice that in the linear case the study of the convergence of the dynamics can be tackled through convergence of the resolvents. In the presence of a nonlinear term this approach fails and for this reason we study the problem by looking directly at the time dependent formulation. We stress that our approach does not need any previous information on the convergence of the linear dynamics, which is obtained instead as an independent (and, to the best of our knowledge, new) byproduct of our proof.\\
A second remark is that one of the main difficulties of the analysis consists in the fact that the scaled problem is well posed in $H^2(\RE^3)$ (see Section \ref{s:limprob}) while the limit problem is well posed in the nonlinear operator domain $\DD$, which is a different (and larger) space including singularities. A first consequence is the need to tune the approximation of the initial datum of the limit problem through a sequence of initial data of the scaled problem, i.e., the scaling acts not only on the form of the nonlinearity, but also on the solution spaces. We will choose a sequence of initial data of the form
\begin{equation}\label{psi0ve}
\psi_0^\ve = \phi_0 + q_0 \rho^\ve * G \in H^2(\RE^3).
\end{equation}
Any element
$
\psi_0 = \phi_0 +q_0 G \in \DD
$
can be approximated in $L^2(\RE^3)$ in this way. Notice that in the approximated initial datum the regular part $\phi_0$ is fixed and the only $\ve$-dependent  term is the one which reconstructs the singular part. Other choices could be done, but we content ourselves with this one.\\
\n
After these premises we state the following theorem which is our main result
\begin{theorem}\label{t:main}
Assume $\gamma\geq 0$ and $\mu\geq 0$ or $\gamma<0$ and $0\leq\mu<1$. Let $\psi(t)$ be the solution of the limit problem \eqref{limiteq} with $\psi_0 =\phi_0+q_0G\in \DD$, $\psi^{\ve}(t)$ be the solution of the approximating problem \eqref{appeq} with $\psi_0^\ve=\phi_0+q_0\rho^{\ve}*G$ as in \eqref{psi0ve} and fix $T>0$. Then
\begin{equation}\label{main}\sup_{t\in[0,T]} \|\psi^\ve(t) - \psi(t)\|_{L^2(\RE^3)} \leq C \ve^\de \end{equation}
where $C$ is a positive constant and $0\leq\de<1/4\ .$
\end{theorem}
We now comment about some aspects of the result and on its proof, which involve non standard procedures and techniques.\\
As quoted before, problem \eqref{limiteq} was studied, in its weak form, in \cite{at2}. There the authors focused on the formulation of the problem in the \emph{energy domain}, i.e., for initial data of the form $\psi(0)=\phi_0+q_0G \in L^2(\RE^3)$, with $\phi_0\in \dot H^1(\RE^3)$. For these initial data and under certain assumptions on the  parameters $\mu$ and $\gamma$,  it is proved that the problem \eqref{limiteq} admits a unique and global in time weak solution of the form $ \phi(t) + q(t) G \in L^2(\RE^3)$, with $\phi(t)\in \dot H^1(\RE^3)$. 
In the present paper we concentrate on {\it strong solutions}, i.e., solutions in the operator domain. This additional smoothness is needed to guarantee the validity of several estimates used in the comparison of the scaled and limit flows. The global well-posedness of the limit problem for subcritical nonlinearity $0\leq\mu<1$ in strong form is studied in Section \ref{s:limprob}, see in particular Theorem \ref{t:wp}. 
On the other hand, in Section \ref{s:appprob} it is studied the global well-posedness of the scaled problem in strong form, which holds true for every power nonlinearity $\mu>0$, see Theorem \ref{t:wpscaled}.\\
In order to compare the scaled and the limit flows it is convenient to recast the two problems in  integral form. Let us consider first the scaled problem. We introduce the ``approximated charge''
\begin{equation}
\label{apprcharge}
q^\ve(t) = \frac{\ve}{\ell} ( \rho^\ve,\psi^\ve(t)),
\end{equation}
and denote by $U(t)$  the free Schr\"odinger group  with integral kernel
\begin{equation*}
U(t,x) = \frac{e^{i \frac{|x|^2}{4t}}}{(4\pi i t )^{3/2}}.
\end{equation*}
By Duhamel formula we write the scaled equation as
\begin{equation}
\label{solveps}
\psi^\ve(t )  = U(t)\psi_0^\ve 
+ i  \int_0^t \ds  (U(t-s)\rho^\ve) q^\ve(s) \\ 
- i \gamma  \frac{\ve}\ell  \int_0^t \ds  (U(t-s)\rho^\ve)\left|q^\ve(s)\right|^{2\mu}   q^\ve(s)
\end{equation}
and from this we obtain an equation for the evolution of the approximated charge
\begin{equation}
\label{eqapprchargeeq}
\frac\ell\ve q^\ve(t ) = (\rho^\ve, U(t)\psi_0^\ve) 
+ i  \int_0^t \ds (\rho^\ve, U(t-s)\rho^\ve) q^\ve(s) \\ 
- i \gamma  \frac{\ve}\ell  \int_0^t \ds (\rho^\ve, U(t-s)\rho^\ve)\left|q^\ve(s)\right|^{2\mu}   q^\ve(s).
\end{equation}
The limit problem has an equivalent formulation in terms of the (proper) charge as well. Define
\begin{equation}\label{solinhom}
\psi(t)=U(t)\psi_0 + i\int_0^t \ds U(t-s,\cdot)q(s) .
\end{equation}
The r.h.s. of \eqref{solinhom} represents a strong solution of the nonlinear limit problem if and only if the charge $q(t)$ satisfies the nonlinear Volterra integral equation (see for the linear case \cite{SY})
\begin{equation}
\label{limit2}
q(t)+4 \sqrt{\pi i}\, \ga \int_0^t \ds \frac{|q(s)|^{2\mu} q(s)}{\sqrt{t-s}} =  4 \sqrt{\pi i} \int_0^t \ds \frac{(U(s)\psi_0)(0)}{\sqrt{t-s}}.
\end{equation}
Eqs. \eqref{eqapprchargeeq} and \eqref{limit2} however are not the right starting point for the comparison of the two nonlinear dynamics. The reason resides in the following fundamental estimate (proven at the beginning of Section \ref{s:6})
\[
 \sup_{t\in[0,T]} \| \psi^\ve (t) - \psi(t) \|_{L^2} \leqslant c \, \left(  \| \psi^\ve_0 - \psi_0 \|_{L^2}+
\| I^{1/2} (q-q^\ve)\|_{L^{\infty} (0,T)}+  \ve^{1/4} \right).
\]
Here
\[
I^{1/2} f (t) = \int_0^t \dd s\, \f{ f(s)}{\sqrt{t-s} }
\] 
is the well known fractional integral of order $1/2$, or Riemann-Liouville operator, which in the present paper is considered as a linear operator between H\"older spaces or homogeneous Sobolev spaces (see Section \ref{s:2} for essential definitions and properties). According to the above estimate, to achieve the main result one has to control the convergence of the initial datum in $L^2(\RE^3)$, which is easy, and the convergence of  $I^{1/2} (q-q^\ve)$, which is much more delicate and involved.\
To this end we recast the charge equations \eqref{limit2} and \eqref{eqapprchargeeq} in the following way. For the limit charge
\[
I^{1/2} q (t) + 4\pi \sqrt{\pi i } \ga \int_0^t \dd s\, |q(s)|^{2\mu} q(s) = 4\pi \sqrt{\pi i } \ga \int_0^t \dd s\,  (U(s)\psi_0) (\bf{0}).
\]
For the approximated charge
\[
I^{1/2} q^\ve (t) + 4\pi \sqrt{\pi i } \ga \int_0^t \dd s\, |q^\ve(s)|^{2\mu} q^\ve(s) = 4\pi \sqrt{\pi i } \ga \int_0^t \dd s\,  (U(s)\psi_0) ({\bf 0}) + Y^\ve (t).
\]
\noindent
The remainder or source term $ Y^\ve(t)$ is a sum of terms explicitly given in Prop. \ref{p:5.1} and Eq. \eqref{Yve4}.
Starting from
\[
I^{1/2} (q^\ve - q)(t)  +4\pi \sqrt{\pi i}\, \ga \int_0^t \dd s ( |q^\ve(s)|^{2\mu} q^\ve(s) - |q(s)|^{2\mu} q(s) ) ={Y}^\ve(t),
\]
to establish the convergence of $I^{1/2} q^\ve$ to  $I^{1/2} q$ it is sufficient to prove the following estimates on the remainders:
\[
\begin{aligned}
&\|{Y}^{\ve}\|_{L^{\infty}(0,T)} \leq C \ve^{1/2} \\ 
&\| D^{1/2}{Y}^{\ve}(t)\|_{L^1(0,T)} \leq C\ve^{\delta} \ \ \ \ \ \ \  (\delta\in [0,1/4),\ \ \ D^{1/2}=\frac{d}{dt} I^{1/2}).
\end{aligned}
\]
To obtain them we need a priori quantitative bounds on approximated charge and wave function
\[
\begin{aligned}
&\|q^\ve \|_{L^{\infty}} \leqslant c \qquad\qquad \  \sup_t \|  \nabla \phi^\ve (t) \|_{L^2} \leqslant c \\
&\|\dot q^\ve \|_{L^{\infty}} \leqslant c \ve^{-3/2} \qquad \|D^{1/2}  q^\ve \|_{L^{1}} \leqslant c \ve^{-1/2+\de}.
\end{aligned}
\]
The first couple of estimates depends on $\ve$-uniform energy estimates on the solution of the scaled problem, proved in Section \ref{s:appprob}. These are non trivial, in particular in the case of attractive nonlinearity ($\gamma<0$).
They hold only on the range of power nonlinearity $\mu$ and coefficients $\gamma$ in which the limit problem has a global solution ($\gamma\geq 0$ and $\mu\geq 0$ or $\gamma<0$ and $0\leq\mu<1$). The second couple of estimates give a quantitative control in $\ve$ on the rate of divergence of the quantities at the l.h.s. and are deduced from the approximate charge equation in Section \ref{s:5}.\\
We add some possible developments.\\
We first notice that the structure of the proof works well for more general gauge invariant nonlinearities than power-type. Moreover the scaling procedure used here could be adapted to a number of non-autonomous problems such as the ionization problem in dimension three (\cite{CdAFM} and \cite{CdA}) or the time-dependent point-interactions (\cite{CCF17}, \cite{HMN} and \cite{NZ}).  In particular  time dependent point interactions could be considered as a limit of charge transfer models, which have a relevant interest in several different contexts (see, e.g., \cite{Grf1990,NieSf2003,RoScgSf-p,Yj,Zi1997}).
Finally, the same approximation procedure could be extended to different equations with concentrated nonlinearities already studied in the literature, in particular the wave and Klein-Gordon dynamics (see, e.g., \cite{NP,Kopylova16,Kopylova17}).\\
\noindent
In conclusion, for the convenience of the reader, we summarize the content of the following Sections.

\n
In Section \ref{s:2} we  introduce some useful notations involving Fourier transforms and various Sobolev spaces needed in the subsequent analysis and we also prove two technical lemmata. 

\n
In Section \ref{s:limprob} we establish the global well-posedness of the limit problem in the domain $\mathcal D$.

\n
In Section \ref{s:appprob} we prove the global well-posedness of the scaled problem  in $H^2(\RE^3)$, two a priori bounds deduced from the conservation of the energy and the convergence of the initial data. 

\n
In Section \ref{s:5} we prove  a priori bounds for the derivative and the fractional derivative of the charge $q^{\ve}$.

\n
In Section \ref{s:6} we finally prove the convergence of the scaled flow to the limit flow using the estimates obtained in the previous Sections.

\section{Notation and preliminaries \label{s:2}}

From now on we denote vectors by bold letters (this is a departure from the usage in the introduction). 

When no misunderstanding is possible, the modulus $|\xx|$ will be simply denoted by $x$, and the modulus $|\kk|$ by $k$. We indicate with $C$ a generic constant, possibly dependent on parameters, which can change from line to line.

\subsection{Fourier transform and convolutions}
We denote by $\hat \psi$ the spatial Fourier transform of $\psi$ 
\[
\hat \psi(\kk) =
\frac{1}{(2\pi)^{3/2}}\int_{\RE^3} \dx e^{-i\kk\cdot \xx} \psi(\xx)\ .
\]
Recall that the Fourier transform defined in this way is unitary.

The time-Fourier transform of $f$ is denoted by $\FF f$ and defined as
\[
\FF f (\ome) := 
\frac{1}{(2\pi)^{1/2}}\int_{\RE} \dt e^{-i\ome t} f(t)\ .
\]
With these definitions the  Fourier transform of the convolution is \[(\widehat{\psi*\phi})(\kk) = (2\pi)^{3/2} \hat \psi(\kk) \hat \phi(\kk),\] and the time-Fourier transform of the convolution is  \[\FF(f*g)(\omega) = (2\pi)^{1/2} \FF f(\ome) \FF g(\ome).\] 
The Fourier transform of the free unitary group kernel $U(t,\xx)$ is 
\[\hat U(t, \kk) = \frac{e^{-i k^2 t }}{(2\pi)^{3/2}},  \] 
and \[(\widehat{U(t)\psi})(\kk) = e^{-ik^2 t} \hat \psi(\kk).\]
\par\noindent
Note moreover that the Fourier transform of the fundamental solution $G(\xx)= 1/(4\pi|\xx|)$ is $$\hat G(\kk) = \frac{1}{(2\pi)^{3/2}|\kk|^2}\ .$$
\par\noindent
We list several elementary properties of the Fourier transform of the form factor $\rho$. Since $\rho$ is spherically symmetric so is $\hat\rho$, in particular 
\[\hat \rho (k) = \sqrt{\frac{2}\pi} \int_0^\infty \dd x \, x\, \frac{\sin kx }{k} \rho(x). \]
Since $\rho\in \mathcal S(\RE^3)$, we have that  $x^\eta  \rho(x)$ is integrable for all $\eta\geq0$.  Hence,  for all $k>0$ there exists a positive constant $C$ such that   $|\hat \rho(k) -\hat\rho(0)|\leq C k^2$, because $\hat \rho (0) = \sqrt{\frac{2}\pi} \int_0^\infty \dd x \, x^2 \rho(x)$ and   $| \frac{\sin x }{x} - 1| \leq x^2/6!$.  Since $|\hat \rho(k) -\hat\rho(0)|\leq C$  as well, one has 
\begin{equation}\label{working}
 |\hat \rho(k) -\hat\rho(0)|\leq C k^\eta \qquad \forall\; 0\leq\eta\leq2.
\end{equation}
We also note that
\[\hat \rho' (k) = \sqrt{\frac{2}\pi} \int_0^\infty \dd x \, x\, \frac{kx\cos kx - \sin kx  }{k^2} \rho(x). \]
\noindent
Hence,  for all $k>0$, one has  $|\hat\rho'(k)|\leq C k$, because $|\frac{x\cos x - \sin x }{x^2}| \leq x$. Since it is also true that $|\hat\rho'(k)|\leq C$  we have that 
\begin{equation}\label{working2}
 |\hat \rho'(k)|\leq C k^\eta \qquad \forall\; 0\leq\eta\leq1.
\end{equation}
\par\noindent
Finally, the constant $\ell$ defined in Eq. \eqref{M} can be written as 
\begin{equation}\label{M2}
\ell =  \int_{\RE^3} \dk \frac{(\hat \rho(\kk))^2}{|\kk|^2}.
\end{equation}

\subsection{Sobolev spaces}
When writing  norms and inner products in Lebesgue or Sobolev spaces the domain $\RE^3$ will be omitted. Moreover  for norms  in $L^p(\RE^3)$, only the index will be written and it will be omitted 
when it is two. That is $\|\psi\|_{H^2}= \|\psi\|_{H^2(\RE^3)}$, $\|\psi\|_{H^1}= \|\psi\|_{H^1(\RE^3)}$,  $\|\psi\|_p = \|\psi\|_{L^p (\RE^3)}$, $\|\psi\| = \|\psi\|_{L^2 (\RE^3)}$ and $(\psi_1,\psi_2)_{L^2(\RE^3)}=(\psi_1,\psi_2)$.

For any $\nu\geq 0$, we denote by $\dot H^{\nu}(\RE^d)$ the  homogeneous Sobolev space over $\RE^d$, which is the space  of tempered distributions with Fourier transform in $L^1_{loc}(\RE^d)\cap L^2(\RE^d,|\kk|^{2\nu}\dk)$.

We shall make large use of fractional (homogeneous and nonhomogeneous) Sobolev spaces in dimension one. For this reason, we recall several definitions and known results. For more details, the reader is referred to \cite{bahouri}, \cite{hitchhiker} and \cite{adams-fournier}.

For any $-\infty\leq a <b\leq +\infty$ and $\nu\in(0,1)$,  we use the notation  
\[
[f]_{\dot H^\nu(a,b)} :=  \left(\int_{[a,b]^2 } \dd s \dd s' \f{ | f(s) -f(s') |^2  }{ |s-s'|^{1+2\nu} }\right)^{1/2},
\]
which is sometimes  referred to as Gagliardo (semi)norm of $f$. 

We recall that for $\nu\in(0,1)$ there exists a constant $C_\nu$ such that 
\[[f]_{\dot H^\nu(\RE)} = C_\nu \|\FF f\|_{L^2(\RE,|\omega|^{2\nu}\dd \omega)}\]
for any $f\in \dot H^\nu(\RE)$ (\cite{bahouri}, Proposition 1.37). This equality implies  that, for $\nu\in(0,1)$,  one can define $H^\nu(\RE)$ as the space of measurable functions such that the norm 
\[\|f\|_{H^\nu(\RE)} = \|f\|_{L^2(\RE)} + [f]_{\dot H^\nu(\RE)}\]
is finite. 

The space $H^{\nu}(a,b)$, for generic $-\infty\leq a <b\leq +\infty$ and $\nu\in(0,1)$,  is  defined in a similar way, i.e., it is the space of functions for which the norm \[\|f\|_{H^\nu(a,b)} = \|f\|_{L^2(a,b)} + [f]_{\dot H^\nu(a,b)}\] is finite. 

To define the space $H^\nu(a,b)$  for  finite $a$ and $b$, and $\nu>1$ one sets $\nu = m + \sigma$, where $m$ is an integer and $\sigma\in(0,1)$.  Then $H^\nu(a,b)$ is the space of functions such that $f\in H^m(a,b)$ and  $f^{(m)}\in H^\sigma(a,b)$.

We recall that the space  $L^2(a,b)$ can be identified with $H^0(a,b)$ and  $\dot H^0(a,b)$. 

We conclude by recalling that if $f\in \dot H^{\nu}(\RE)$, for $\nu\geq 0$, then $f\in L^2_{loc}(\RE)$, hence $f \in H^\nu(a,b)$ for any finite $a$  and $b$ (see, e.g., \cite[Prop. 1.37]{bahouri}).  

\subsection{Technical lemmata}\
To prove Lemma \ref{l:wpq} in Section \ref{s:limprob}, we shall use tools from Fourier analysis. To this aim we shall need to extend functions on finite intervals  to functions  on $\RE$, preserving at the same time their regularity properties.  To control the regularity of  the  prolonged function we shall use the following lemma. 
 
 \begin{lemma}\label{l:prolong}
 Let $-\infty<a<b<\infty$ and let $f\in H^\nu(a,b)$ with $\nu\geq0$. Define 
 \[\tilde f (s) = \left\{ \begin{aligned} &f(s) \quad && \text{if} \quad s \in [a, b] \\ 
 & 0 && \text{otherwise}  \end{aligned}\right. \]
 \begin{enumerate}[i)]
 \item If  $0\leq\nu <1/2$, then $\tilde f\in H^\nu(\RE)$. 
 \item  If  $1/2 <\nu <3/2$ and $f(a)=f(b)=0$, then $\tilde f\in H^\nu(\RE)$ 
 \end{enumerate}
 In both cases there exists a constant $C$ such that $\|\tilde f \|_{H^{\nu}}\leq C\|f\|_{H^\nu}\ .$
\end{lemma}
\begin{proof} 
The case $\nu=0$ is trivial. For  $0<\nu<1$,  it remains to prove that $\FF \tilde f\in L^2(|\omega|^{2\nu} \dd \omega)$, or equivalently that $[\tilde f]_{\dot H^\nu(\RE)}<\infty$. We note that 
\[[\tilde f]_{\dot H^\nu(\RE)}^2 =  [f]_{\dot H^\nu(a,b)}^2 + \frac{1}{\nu} \int_a^b \frac{|f(s)|^2}{(s-a)^{2\nu}}\dd s + \frac{1}{\nu} \int_a^b \frac{|f(s)|^2}{(b-s)^{2\nu}}\dd s, \]
where we used the definition of $\tilde f$ and the Fubini-Tonelli theorem to exchange the order of integration. \\
In case $i)$ we use  the inequality  
\begin{equation}\label{part}
\int_0^d \frac{1}{s^{2\nu}} \left|h(s) - \frac{1}{d}\int_0^dh(s')\dd s'\right|^2\dd s \leq  C \int_{(0,d)^2} \frac{|h(s) -h(s')|^2}{|s-s'|^{2\nu+1}} \dd s\,\dd s' , 
\end{equation}  
which holds true for any $0<\nu<1/2$ and $0<d<\infty$.  To prove it we adapt an argument from  \cite{KP}, Sec. 5.2. We note first the identity
\begin{equation}\label{ghost} h(s) - \frac{1}{d} \int_0^d h(s') \dd s' = h(s) -\frac1s \int_0^s h(s') \dd s' - \int_s^d \frac1{s'} \left( h(s') -\frac1{s'} \int_0^{s'} h(\eta) \dd \eta\right) \dd s' \end{equation}
 which can be checked by integrating by parts in $s'$ the latter integral at the r.h.s. of the equality. Next we use the Hardy inequality 
\[\int_0^d \frac{1}{s^{2\nu}}\left|\int_s^d g(s') \dd s'\right|^2 \leq C \int_0^d \frac{|g(s)|^2}{s^{2\nu-2}} \dd s \]
which holds true for any $0<\nu<1/2$ and $0<d<\infty$, see, e.g., \cite[Eq. (0.23)]{KP}.  We set $g(s) = \frac1s \left( h(s) -\frac1s \int_0^s h(\eta) \dd \eta\right)  $ in Eq. \eqref{ghost} and obtain 
\[ 
\int_0^d \frac{1}{s^{2\nu}} \left|h(s) - \frac{1}{d} \int_0^d h(s') \dd s'\right|^2 \dd s \leq  C \int_0^d \frac{1}{s^{2\nu}} \left|h(s) - \frac{1}{s} \int_0^s h(s') \dd s'\right|^2 \dd s.
\]
To conclude the proof of inequality \eqref{part} we use a straightforward adaptation of the argument in the proof of Th. 5.9 in \cite{KP}, which gives  
\begin{multline*}\int_0^d \frac{1}{s^{2\nu}} \left|h(s) - \frac{1}{s} \int_0^s h(s') \dd s'\right|^2 \dd s = 
\int_0^d \frac{1}{s^{2\nu}} \left| \frac{1}{s} \int_0^s \big(h(s) - h(s')\big) \dd s'\right|^2 \dd s  \\ \leq \int_0^d \frac{1}{s^{2\nu+1}} \int_0^s |h(s) - h(s')|^2 \dd s' \dd s \leq \int_0^d \int_0^d  \frac{ |h(s) - h(s')|^2}{|s-s'|^{2\nu+1}} \ \dd s' \,\dd s. 
\end{multline*}
By using inequality \eqref{part} with  $h(s) = f(s+a)$ and $d =b-a$, we have that  
 \begin{multline*}
  \int_a^b \frac{|f(s)|^2}{(s-a)^{2\nu}}\dd s \\ 
   \leq 2\int_a^b \frac{1}{(s-a)^{2\nu}} \left|  f(s) - \frac{1}{b-a}\int_a^b f(s') \dd s'\right|^2 \ \dd s   +2  \left| \frac{1}{b-a}\int_a^b f(s') \dd s'\right|^2  \int_a^b \frac{1}{(s-a)^{2\nu}} \ \dd s  \\ 
   \leq C[f]^2_{\dot H^\nu(a,b)} + \frac{2\|f\|^2_{L^2(a,b)}}{(1-2\nu)(b-a)^{2\nu}} \leq     C  \|f\|_{ H^\nu(a,b)}^2.
 \end{multline*}
 A similar argument gives 
\[
 \int_a^b \frac{|f(s)|^2}{(b-s)^{2\nu}}\dd s  \leq C  \|f\|_{ H^\nu(a,b)}^2. 
\]

In case $ii)$ we use the Hardy inequality 
\[ \int_a^b |f(s)|^2 \left(\frac{1}{(s-a)^{2\nu}}+\frac{1}{(b-s)^{2\nu}}\right)\dd s \leq C \int_{(a,b)^2} \frac{|f(s)-f(s')|^2}{|s-s'|^{1+2\nu}}  \dd s \, \dd s' \] which holds true for any $f\in H^{\nu}(a,b)$, for $1/2<\nu<1$, with $f(a) = f(b) =0$. For the proof we refer to \cite{SL10}. There, the inequality is proved for functions in $C_c^{\infty}(a,b)$ (infinitely differentiable and compactly supported in $(a,b)$). By a density argument, one infers that it holds true for any function $f\in H^{\nu}(a,b)$, with $f(a) = f(b) =0$ (see also \cite[Th. 2.4.8]{Sthesis}). This implies the proof of case $ii)$, $ [\tilde f]_{\dot H^\nu(\RE)} \leq C \|f\|_{H^\nu(a,b)}$ with $1/2<\nu<1$.  To conclude the proof of the Lemma it remains to show that the case $ii)$ holds true also for $1\leq \nu <3/2$. By Th. 5.29 in \cite{adams-fournier} the statement is true for $\nu = 1$, hence it is true also for $1<\nu<3/2$, by applying case  $i)$ to $\dot f$. 
\end{proof}
\begin{remark}
The argument in the above proof fails for $\nu=1/2$, where the Hardy inequality fails (see \cite{KP}).
\end{remark}

It will be occasionally useful the embedding of homogeneous fractional Sobolev spaces in H\"older spaces. In particular (see \cite{bahouri}, Theorem 1.50) for $\nu-1/2 >0$ and not an integer, $\dot H^\nu(\RE)$ is embedded in $C^{[\nu-1/2],\ \nu-1/2-[\nu-1/2]}(\RE)$.

To show that the charge $q$ enjoys   the regularity properties stated in Lemma \ref{l:wpq} in Section \ref{s:limprob}, we shall bootstrap on Eq. \eqref{limit2}. To this aim we define the operator  $I^{1/2}$ and its (up to a factor $\pi$) inverse $D^{1/2}$ as  
\[
I^{1/2}f(t) : = \int_0^t \ds \frac{f(s) }{\sqrt{t-s}} , \qquad D^{1/2} f(t):= \frac{d}{dt} \int_0^t \dd s \frac{f(s)}{\sqrt{t-s}}.  
\]
We note that with this definition $D^{1/2}$ is proportional to the inverse of $I^{1/2}$, i.e., $I^{1/2}D^{1/2} f  = D^{1/2} I^{1/2}f = \pi f $.  In  the following lemma, we recall the regularizing properties of $I^{1/2}$ as an operator    between homogeneous Sobolev spaces, for the proof we refer to \cite[Lemma 3]{at1}.
\begin{lemma}\label{l:reg} Let  $\nu \geq0$ and $T>0$. Assume that  $f \in \dot H^\nu (\RE)$ and has  support in $[0,T]$.  Then $I^{1/2}f\in\dot H^{\nu+1/2}(\RE)$.
\end{lemma}
\begin{remark}\label{r:2.4}
The mapping properties of $I^{1/2}$ between H\"older spaces (\cite{GF}, Theorem 4.2.1) will be useful: If $0<\beta<1/2$, $f\in C^{0,\beta}([0,T])$ and $f(0)=0$ then $I^{1/2}f\in C^{0,\beta+1/2}([0,T])\ .$
\end{remark}
We end this Section with the following identity that will be used later on in the proof of the main theorem and which holds for $g\in L^{\infty}((0,T)),$ $f\in H^{\nu}((0,T))$ with $\nu>1/2$ and $f(0)=0:$  
\begin{equation}\label{this}
\int_0^t \dd s\, g(s) f(t-s) = \frac{1}{\pi}\int_0^t \dd s \, I^{1/2} g (s) \frac{d}{dt} I^{1/2}f(t-s).  
\end{equation}
The identity follows using integration by parts. 

\section{The limit problem}\label{s:limprob}

In this section, making use of the integral formulation  \eqref{solinhom}-\eqref{limit2},  we prove that the limit problem  \eqref{limiteq} is well posed and admits global solutions in case of repulsive nonlinearities or mildly attractive nonlinearities (see Th. \ref{t:wp} and Rem.  \ref{r:3.6} below). 
First we prove local existence  in $\DD$, then, using the conserved quantities, namely mass and energy, we prove that the solution is global.
This evolution problem was analyzed also in \cite{at2}. Compared to  \cite{at2}, we need to prove the existence 
of more regular solutions. The increase in regularity requires new technical
ingredients such as the fractional Hardy inequalities used in Lemma \ref{l:prolong}.
The first step in proving local existence is the analysis of \eqref{limit2}.
To this aim we recall the following proposition (see \cite{at2}), which establishes the Sobolev regularity of $(U(\cdot)\phi_0)(\oo)$.
\begin{proposition} \label{p:harvey}
Let $\phi\in \dot H^2(\RE^3)$, then $(U(\cdot)\phi)(\oo) \in \dot H^{3/4}(\RE)$.
\end{proposition}
\noindent The  proof of Prop. \ref{p:harvey} follows without modifications the argument used in the proof of Th. 4 in  \cite{at2}. 
\noindent Now we prove local existence for Eq. \eqref{limit2}. 
 \begin{proposition}\label{l:wpq}
 Let $\psi_0 \in \DD$, then there exists $T_0>0$ such that equation \eqref{limit2} admits a unique solution $q\in  H^{5/4}(0,T_0)$. 
 \end{proposition}
\begin{proof}
First we cast \eqref{limit2} in a more convenient form.
Since $\psi_0\in \DD$, 
\begin{equation}\label{anta}
4 \sqrt{\pi i} \int_0^t \ds \frac{(U(s)\psi_0)(\oo)}{\sqrt{t-s}}  = 4 \sqrt{\pi i} \int_0^t \ds \frac{(U(s)\phi_0)(\oo)}{\sqrt{t-s}}  + q_0 .
\end{equation}
Indeed one has  
\begin{multline*}
4 \sqrt{\pi i} \int_0^t \ds \frac{(U(s)G)(\oo)}{\sqrt{t-s}} = 4 \sqrt{\pi i} \int_0^t \ds \frac1{\sqrt{t-s}}  \int_{\RE^3} \dk \hat U(s,\kk) \hat G(\kk) \\ 
=  4 \sqrt{\pi i} \int_0^t \ds \frac1{\sqrt{t-s}}  \int_{\RE^3} \dk \frac{e^{-i|\kk|^2s}}{(2\pi)^3|\kk|^2} 
=  \frac{ 2\sqrt{ i}}{\pi^{3/2}} \int_0^t \ds \frac1{\sqrt{t-s}}  \int_0^\infty \dd k\, e^{-i k^2s} 
=1,
\end{multline*}
using the Fresnel integral
\begin{equation} \label{identities1}
\int_0^\infty \dd k\, e^{-i k^2 s} = \frac{1}{2\sqrt s} \sqrt{\frac{\pi}{ i}} 
\end{equation}
and the elementary identity
\begin{equation}\label{identities2}
  \int_0^t\dd s\, \frac{1}{\sqrt{t-s}} \frac{1}{\sqrt{s}}  = \pi .
\end{equation}
Adding and subtracting a suitable quantity we rewrite \eqref{anta} as
\[
q(t)-q_0 +4 \sqrt{\pi i}\, \ga \int_0^t \ds \frac{|q(s)|^{2\mu} q(s)-|q_0|^{2\mu}q_0}{\sqrt{t-s}} =  4 \sqrt{\pi i} \int_0^t \ds \frac{(U(s)\phi_0)(\oo)-\gamma|q_0|^{2\mu}q_0}{\sqrt{t-s}},
\]
or in a more compact form
\begin{equation}\label{limit2b2}
(q-q_0) +4 \sqrt{\pi i}\, \ga\,  I^{1/2} ( |q|^{2\mu} q-|q_0|^{2\mu}q_0 ) = 4 \sqrt{\pi i}\, I^{1/2} f
\end{equation}
where we have denoted $(U(t)\phi_0)(\oo)-\gamma|q_0|^{2\mu}q_0$ by $f(t)$ for the sake of brevity. Notice that $f$ is a continuous function vanishing at $t=0$ due to Prop.  \ref{p:harvey},
Sobolev embedding theorem and boundary condition. More precisely (see Rem. \ref{r:2.4}), $f\in C^{0,1/4}(\RE^+)$. Then $I^{1/2} f$ is continuous too; in fact $I^{1/2}f\in C^{0,3/4}(\RE^+)\ .$

It is straightforward to prove that \eqref{limit2b2} admits a solution $q\in C([0,T_0] )$ for sufficiently small $T_0$ by a contraction argument (see \cite{Miller}).
Due to the continuation properties of solutions of nonlinear Volterra integral equations, there exists a 
maximal time of existence $T^*$ and the following alternative holds: either $T^*=+\infty$, that is the solution is global, or $\limsup_{t\to T^*} |q(t)|= +\infty$.

Now we prove that the solution has $H^{5/4}(0,T_0)$ regularity by means of a bootstrap argument.
Let us analyze the source term $I^{1/2} f$. As already remarked, $f(0)=0$ because of the boundary condition in $\DD$.
Let us take the restriction of $f$ to $[0, T_0]$ and extend it to $[0,2T_0]$ by reflection w.r.t. $t=T_0$, i.e., 
\[
f_s(t) = \left\{\begin{aligned} &f(t) \quad && t\in [0,T_0] \\ 
& f(2T_0 - t)&& t\in(T_0 ,2T_0]\end{aligned}\right.
\]
Moreover we extend $f_s$ to a function $\tilde f$ defined on the real line in the following way: 
\[
\tilde f(t) = \left\{\begin{aligned}&f_s(t) \quad && t\in (0,2T_0) \\ 
& 0 && \text{otherwise}\end{aligned}\right.
\]
It is well know, see \cite{hitchhiker} , that the extension by reflection is continuous w.r.t. fractional Sobolev norm, that is:
\[
\|f_s\|_{L^2(0,2T_0)}^2 = 2\|f\|_{L^2(0,T_0)}^2\ ,  \qquad
 [f_s]^2_{\dot H^{3/4}(0,2T_0)}  \leq 4[f]^2_{\dot H^{3/4}(0,T_0)}. 
\] 
Since $f_s(0) =f_s(2T_0) = 0$, by Lemma \ref{l:prolong}-\emph{ii)},  we conclude that  $\tilde f \in H^{3/4}(\RE)$ and is compactly supported. 
Then by Lemma \ref{l:reg} we have $I^{1/2} \tilde f \in  \dot H^{5/4}(\RE)$  
and  $I^{1/2} \tilde f \in L^2(0,T_0)$ hence $I^{1/2} \tilde f \in   H^{5/4}(0,T_0)$.
Notice that by construction for any $t\in [0,T_0 ]$ we have 
\[
\int_0^t \dd s \frac{f(s)}{\sqrt{t-s}}   = \int_0^t \dd s \frac{\tilde f(s)}{\sqrt{t-s}},
\]
and then $I^{1/2} f \in H^{5/4} (0, T_0) $.

Next we start to bootstrap. Recall that $q$ is continuous and bounded in $[0,T_0]$. Repeating the previous argument, 
using this time Lemma \ref{l:prolong}-\emph{i)} and Lemma \ref{l:reg}, we have 
$  I^{1/2} ( |q|^{2\mu} q-|q_0|^{2\mu}q_0 ) \in H^{1/2}(0, T_0)$ and then $q\in  H^{1/2}(0, T_0)$ by \eqref{limit2b2} since $I^{1/2} f \in H^{5/4} (0, T_0) $.

We can not immediately iterate the argument since in Lemma \ref{l:prolong} the case $\nu=1/2$ is not included.
However it is enough to notice  that $q\in  H^{1/4}(0, T_0)$ as well and we can repeat the argument. 
Since  $q\in  H^{1/4}(0, T_0)$ and it is bounded then, by the inequality
\begin{equation} \label{tame}
\big||q(s)|^{2\mu}q(s) - |q(s')|^{2\mu}q(s')\big| \leq C \big||q(s)|^{2\mu}+ |q(s')|^{2\mu}\big| \, |q(s)-q(s')| \leq C \| q\|^{2\mu}_{L^{\infty}(0,T_0) } |q(s)-q(s')|
\end{equation}
 we have that  $\||q|^{2\mu}q\|_{H^{1/4}(0,T_0)} \leq C\|q\|^{2\mu}_{L^\infty(0,T_0)} \, \|q\|_{H^{1/4}(0,T_0)}$, from which  $|q|^{2\mu} q \in H^{1/4} (0,T_0)$.
 Again, by the extension argument above  and Lemma \ref{l:prolong}-\emph{i)}, we have 
$  I^{1/2} ( |q|^{2\mu} q-|q_0|^{2\mu}q_0 ) \in H^{3/4}(0, T_0)$ and then $q  \in H^{3/4}(0, T_0)$.
In the final step of the bootstrap we notice that $|q|^{2\mu} q \in H^{3/4} (0,T_0)$ by \eqref{tame} and
 use Lemma \ref{l:prolong}-\emph{ii)} to prove 
$  I^{1/2} ( |q|^{2\mu} q-|q_0|^{2\mu}q_0 ) \in H^{5/4}(0, T_0)$; it follows $q  \in H^{5/4}(0, T_0)$.
\end{proof}

Next we study the local well-posedness properties of \eqref{limiteq}. 
 \begin{proposition}\label{l:locwp}
 Let $\psi_0 \in \DD$, then there exists $T_0$ such that problem  \eqref{limiteq} admits a unique solution $\psi(t)\in \DD$ for all $t\in[0,T_0]$. Moreover the map $t \leadsto \psi(t)$  belongs to $C([0,T_0],\DD)\cap C^1([0,T_0],L^2(\RE^3))$. 
 \end{proposition}
\begin{proof}Let $[0,T_0]$ and $q\in H^{5/4}(0,T_0)$, be the time interval and the solution of Eq.\eqref{limit2} given by Proposition \ref{l:wpq}. 
Let $\psi(t)$ be represented in the form \eqref{solinhom} and define $\phi(t)$ by \[\psi(t) = \phi(t) + q(t)G.\]
We claim that $\phi(t)\in \dot H^2(\RE) $ for $\forall t\in[0,T_0]$. To prove this assertion we need a suitable representation of $\phi(t)$.
By writing \eqref{solinhom}  in Fourier space and integrating by parts, we have that
\begin{align}
\hat \phi(t,\kk) &= \hat  \psi(t,\kk) - \frac{q(t)}{(2\pi)^{3/2}k^2}  \nonumber \\
&= e^{-ik^2t} \hat\phi_0(\kk) + \frac{q_0 e^{-ik^2t}}{(2\pi)^{3/2}k^2}     - \frac{q(t)}{(2\pi)^{3/2}k^2} + \frac{i}{(2\pi)^{3/2}} \int_0^t   e^{-ik^2(t-s)} q(s)\dd s \nonumber\\
&= e^{-ik^2t} \hat\phi_0(\kk) - \frac{1}{(2\pi)^{3/2}k^2} \int_0^t   e^{-ik^2(t-s)} \dot q(s)\dd s \equiv\hat \phi_1(t,\kk) + \hat \phi_2(t,\kk) \label{hundred}
\end{align}
We have $\phi_1 (t) \in \dot H^2 (\RE)$ since the free Schr\"odinger evolution preserves all the homogeneous Sobolev spaces.
 We need to prove that $ \hat \phi_2 (t)\in L^1_{loc}(\RE^3)\cap L^2(\RE^3,|\kk|^{4}\dk)$. 
Local integrability follows from the pointwise estimate $|\phi_2 (t, \kk)|\leq c k^{-2} \|  \dot q \|_{L^2(0,T_0)}$. Moreover knowing that $ \dot q \in H^{1/4}(0,T_0)$; by Lemma \ref{l:prolong}-\emph{i)}, the function 
\[\widetilde{\dot{q}_t}(s) = \left\{ \begin{aligned}&\dot q(s)\qquad && s\in[0,t]\\ 
&0&& \text{otherwise} \end{aligned}\right.\]
is in $\dot H^{1/4}(\RE)$ for all $t\in[0,T_0]$.  Then we have 
\begin{align}
\|\hat \phi_2(t)\|^2_{ L^2(\RE^3,|\kk|^{4}\dk)} &
= 4\pi \int_0^{\infty}|\hat \phi_2(t,k)|^2 k^6 \dd k  = \frac1\pi \int_0^\infty \left|  \int_0^t e^{-i k^2 s} \dot q (s) \dd s   \right|^2 k^2 \dd k \nonumber \\
&= \frac1\pi \int_0^\infty \left|  \int_\RE e^{-i k^2 s}    \tilde{ \dot q}_t (s) \dd s                        \right|^2 k^2 \dd k  
= \frac1\pi \int_0^\infty \left|     \int_\RE e^{-i \ome s}    \tilde{ \dot q}_t (s) \dd s               \right|^2 |\omega|^{1/2} \dd \omega\nonumber\\
&\leq 2 \int_\RE \left|     \FF  \tilde{ \dot q}_t (\ome)            \right|^2 |\omega|^{1/2} \dd \omega 
= 2[ \tilde{ \dot q}_t ]^2_{\dot H^{1/4} (\RE)} \leq C  \|\dot q \|^2_{H^{1/4} (0,T_0)} \label{prison}
\end{align}
which concludes the proof that $\phi(t)\in \dot H^2(\RE)$.  

To prove that $\psi(t)\in\DD$, it remains to show that the ``boundary condition'' $\phi(t, \oo) - \gamma |q(t)|^{2\mu}q(t)=0$ is satisfied. 
Integrating \eqref{hundred} and using the Fresnel integral \eqref{identities1}, we have 
\begin{equation} \label{vite}
\phi(t, \oo)= \f{1}{(2\pi)^{3/2} } \int _{\RE^3}\hat\phi (t, \kk) \, \dd \kk
=( U(t)\phi_0) (\oo) - \frac{1}{4\pi \sqrt{\pi i }} \int_0^t \frac{\dot q (s)}{\sqrt{t-s} }\dd s 
\end{equation}
Moreover, since $I^{1/2} f = 0 $ implies $f=0$ for continuous $f$, it is sufficient to prove that
\begin{equation} \label{cerniera}
\int_0^t \frac{\dd s}{\sqrt{t-s} } \lf( \phi(s, \oo) - \gamma |q(s)|^{2\mu}q(s) \ri)=0\ .
\end{equation}
Substituting \eqref{vite}, exchanging integrals and using \eqref{identities2}, we see that \eqref{cerniera} is equivalent to \eqref{limit2} and therefore
$\psi(t) \in \DD$.

Notice that the Sobolev regularity 5/4 of $q$ obtained in Proposition \ref{l:wpq} is optimal: a weaker result would not suffice
to prove the invariance of $\DD$ under the dynamics.

Now we prove that the map $t \leadsto\psi(t)$ belongs to $C([0,T_0], \DD)$. Since $q(t) $ is a continuous function, it is enough to prove that 
$t \leadsto\phi(t)$ is a continuous $\dot H^2$-valued function. This holds true for $t\leadsto\phi_1(t)$ due to the properties of the free Schr\"odinger group.
By \eqref{prison} it is sufficient to prove that
\[
\hat \varphi (t, \kk) = \int_0^t e^{-i k^2 s} \dot q (s) \dd s
\] 
is an $L^2(\RE^3)$-valued continuous function, that is
\[
\lim_{\de\to 0}\|\hat \varphi(t+\de)-\hat \varphi(t)\|=0.
\]
This amounts to
\[
\lim_{\de\to 0} \int_0^\infty \lf| \int_t^{t+\de} e^{-i k^2 s} \dot q (s) \dd s\ri|^2 k^2 \, \dd k=0.
\]
The proof is almost contained in the previous arguments but we write it explicitly for the sake of completeness.
As in \eqref{prison}, we have
\begin{align*}
\int_0^\infty \lf| \int_t^{t+\de} e^{-i k^2 s} \dot q (s) \dd s\ri|^2 k^2 \, \dd k 
&= \int_0^\infty \lf| \int_\RE e^{-i k^2 s} \chi_{[t,t+\de]}(s)\, \dot q (s) \dd s\ri|^2 k^2 \, \dd k\nonumber  \\
&= \int_0^\infty \lf| \FF( \chi_{[t,t+\de]}\, \dot q )(\ome) \ri|^2 |\ome|^{1/2} \, \dd\ome  \nonumber \\
&\leq \int_\RE \lf| \FF( \chi_{[t,t+\de]}\, \dot q )(\ome) \ri|^2 |\ome|^{1/2} \, \dd\ome \nonumber  \\
&= [ \chi_{[t,t+\de]}\, \dot q ]_{\dot H^{1/4}(\RE) }^2 \leq C [ \dot q ]_{\dot H^{1/4}(t,t+\de) }^2 \ .
\end{align*}
where in the last step, we have used Lemma \ref{l:prolong}-\emph{ i)}.
Since 
\[
\lim_{\de\to 0} [ \dot q ]_{H^{1/4}(t,t+\de) }^2  = \lim_{\de\to 0}\int_t^{t+\de} \frac{ |\dot q (s) -\dot q (s')|^2 }{|s-s'|^{3/2} }\dd s \, \dd s' =0
\]
by the absolute continuity of the integral and $\dot q \in H^{1/4}(0,T_0)$, continuity of $\phi(t)$ is proved.

To prove that the map  $t\leadsto\psi(t)$ belongs to $C^1([0,T_0], L^2(\RE^3))$, 
we note that by taking  the derivative of the Fourier transform of Eq. \eqref{solinhom}, see also \eqref{hundred}, one obtains
\[
\pd{}{t} \widehat{\psi}(t,\kk) = -ik^2 e^{-ik^2t} \hat\phi_0(\kk) + \frac{i}{(2\pi)^{3/2}}\hat \varphi(t,\kk). 
\]
Hence the statement follows by the same argument used above.
\end{proof}
The limit problem admits two conserved quantities: mass and energy. The former is given by the $L^2(\RE^3)$ norm while
the latter is defined in the following way:
\begin{equation} \label{energyb}
E[\psi] = \| \nabla \phi \| ^2+ \frac{\ga}{\mu+1} |q|^{2\mu+2}
\end{equation}
The energy space $\EE$ is defined as follows: 
\[
\EE:= \left\{ \psi\in L^2 (\RE^3)| \, \psi(\xx)= \phi(\xx) +\f{q}{4\pi|\xx|} ;\, \phi \in \dot H^1 (\RE^3), \, q\in \CO\right\},
\]
Comparing  $\EE$ with $\DD$ we
remark the different  smoothness of the regular part $\phi$ and the absence of the boundary condition. In \cite{at2} the well-posedness in $\EE$ was proved.

\begin{remark}\label{r:EEDD} Notice that, even if homogeneous Sobolev spaces do not form a scale, $\psi\in \DD $ has finite energy due to the $L^2$-constraint on $\psi$.
Indeed let $\psi=\phi+q G\in \DD$; in order to prove that $E[\psi]<\infty$ it is sufficient to consider bounded $k$ and to prove that 
\[
\int_{k<1} | k\hat \phi(\kk)|^2 \dd \kk <+\infty\ .
\]
Since $k \hat \phi(\kk) = k\hat \psi(\kk)- \frac{q}{(2\pi)^{3/2} k}$ then
\[
\int_{k<1} | k\hat \phi(\kk)|^2 \dd \kk \leq c  \int_{k<1} | k\hat \psi(\kk)|^2 \dd \kk + C |q|^2 \int_{k<1} \frac{1}{k^2} \dd \kk\leq C \|\psi\|^2 +|q|^2\ .
\]
\end{remark}
\vskip10pt

\begin{proposition}\label{p:conservation}
Let $\psi(t)$ be the local solution constructed in Proposition \ref{l:locwp}; then we have
\[
\|\psi(t)\|=\| \psi_0\|\ , \qquad \qquad E[\psi(t)]=E[\psi_0] \quad \quad \forall t\in [0,T_0]\ .
\]
\end{proposition}
The proof is as in \cite{at2}.

\vspace{.5cm}

We are now ready to prove the global well-posedness in the domain $\DD$. 
\begin{theorem}\label{t:wp} 
Let  $\psi_0\in\DD,$  $\gamma\geq 0$ and $\mu\geq 0$ or $\gamma<0$ and $0\leq\mu<1$; then the solution of \eqref{solinhom} - \eqref{limit2} is global in time in $\DD$,  i.e., for any $T>0$ there exists a unique solution $\psi(t)\in \DD$ for all $t\in[0,T]$.  Moreover the map $t\leadsto \psi(t)$  belongs to $C([0,T],\DD)\cap C^1([0,T],L^2(\RE^3))$. 
\end{theorem}
\begin{proof}
The strategy of the proof is standard.
The maximal time of existence $T^*$ of equation \eqref{limit2} is infinite since the smallness of $T_0$ does not play any role in the proof of Proposition
\ref{l:locwp} and, due to the blow up alternative, $q(t)$ stays bounded uniformly in time.
For repulsive nonlinearities this last claim is obvious due energy conservation. For attractive nonlinearities we take into account the following estimate (see \cite{at2})
\[
 |q(t)|^2\leq C  \|\psi(t)\| \|\nabla \phi(t)\|\leq C   \|\nabla \phi(t)\|;
 \]
uniform boundedness of $q$ follows from conservation of energy under the hypothesis $0\leq\mu<1$ by a straightforward argument.
\end{proof}

\begin{remark}\label{r:3.6}
The solution of  \eqref{solinhom} - \eqref{limit2}  in  Th. \ref{t:wp} is also a solution of problem  \eqref{limiteq}.
\end{remark}

\section{The scaled problem}\label{s:appprob}

Here we prove the global well-posedness of the approximating problem, the convergence of the initial datum $\psi_0^\ve$ and some a priori estimates derived from energy conservation.\\
The integral form of the scaled problem \eqref{appeq} - \eqref{appid} is 
\begin{equation}\label{weakap3d}
\psi^\ve(t,\xx ) = (U(t)\psi^\ve_0)(\xx )  - i \frac{\ve}{\ell} \int_0^t ds (U(t-s)  \rho^\ve)(\xx ) \left(-1 + \gamma \frac{\ve^{1+2\mu} |(\rho^\ve,\psi^\ve(s))|^{2\mu}}{\ell^{2\mu+1}}\right)(\rho^\ve,\psi^\ve(s)). 
 \end{equation}
 The following Theorem gives local well-posedness for strong solutions of problem \eqref{weakap3d} for any fixed $\ve$.
 
 \begin{theorem}\label{t:app-wp}
Fix  $\ve>0$ and  let   $\psi^{\ve}_0\in H^2(\RE^3)$. Then  there exists $T^{*}>0$ such that Eq. \eqref{weakap3d} has a unique solution $\psi^\ve \in C((0,T^{*}), H^2(\RE^3))\cap C^{1}((0,T^{*}), L^2(\RE^3)).$  \\
On the existence interval $[0,T^{*})$, mass $M[\psi]$ and energy $E^\ve[\psi]$ given by
\[
M[\psi]= \| \psi\|^2 \qquad \text{and} \qquad E^\ve[\psi ] =   \| \nabla \psi\|^2 - \frac{\ve}{\ell} |(\rho^\ve, \psi )|^2 +\frac{ \ga}{\mu+1}\frac{ \ve^{2+2\mu}}{\ell^{2+2\mu}} |(\rho^\ve,\psi )|^{2+2\mu}
\]
are finite and conserved along the flow.
\end{theorem}
\begin{proof}
The proof follows from a standard application of Banach fixed point theorem to the map $K:\overline{B_X(\psi_0^\ve, M)}\rightarrow \overline{B_X(\psi_0^\ve, M)}$, where $X=C((0,T^{*}), H^2(\RE^3))$, $\overline{B_X(\psi_0^\ve, M)}$ is the closed ball in $X$ of radius $M$ and center $\psi_0^\ve$ and

$$
K\psi^\ve(\xx)=(U(t)\psi^\ve_0)(\xx )  - i \frac{\ve}{\ell} \int_0^t ds (U(t-s)  \rho^\ve)(\xx ) \left(-1 + \gamma \frac{\ve^{1+2\mu} |(\rho^\ve,\psi^\ve(s))|^{2\mu}}{\ell^{2\mu+1}}\right)(\rho^\ve,\psi^\ve(s)).$$
\par\noindent
A direct bootstrap argument shows then that $\psi^\ve\in C^{1}((0,T^{*}), L^2(\RE^3))$.\\
Conservation of mass and energy are a consequence of standard computations as well, taking in account the regularity of the solution $\psi^\ve$.

\end{proof}
\begin{remark}
Due to the regularity shown in the previous result, the strong solution of the integral equation \eqref{weakap3d} is also a solution of the differential equation \eqref{appeq}.
\end{remark}
In the following theorem we state the global well-posedness of the scaled problem for initial data in $H^2(\RE^3)$. 
\begin{theorem}\label{t:wpscaled}
Fix $\ve>0$, and $\psi^{\ve}_0\in H^2(\RE^3)$. Then the unique solution $\psi^\ve(t) $ of Eq. \eqref{weakap3d}  is global in time, i.e., for any $T>0$ there exists a unique solution $\psi^\ve \in C((0,T), H^2(\RE^3))\cap C^{1}((0,T), L^2(\RE^3)). $  
\end{theorem}
\begin{proof}
Let us call $T^*$  the existence time obtained in the local existence theorem. Using easy estimates and conservation of $L^2$ norm of the local solution one shows that for every $t\in (0,T^*)$ the following inequality holds true
$$
\|\psi(t) \|_{H^2} \leq C(\rho^{\ve}, \|\psi(t)\|^{2\mu})\|\psi^{\ve}_0\|_{H^2}\ T^*=C(\rho^{\ve}, \|\psi_0\|^{2\mu})\|\psi^{\ve}_0\|_{H^2}\ T^* \ ,
$$
where $C(\rho^{\ve}, \|\cdot\|^{2\mu})$ denotes a constant which may depend on $\rho^{\ve}$ and $\|\cdot\|^{2\mu}$. This allows for the continuation of the $H^2$ solution beyond $T^*$, obtaining a similar estimate with the same constant, and the iteration of the continuation gives existence for all $t\in \RE\ .$
\end{proof}
Now we introduce the regularized charge and rewrite the total energy of the scaled problem in a form analogous to the limit energy, see \eqref{energyb}.
For any function $\psi\in L^2(\RE^3)$ we set 
\begin{equation}\label{qve}
q^\ve = \frac{\ve}{\ell} (\rho^\ve, \psi),
\end{equation}
and define
\[
\phi^\ve = \psi -  q^\ve \rho^\ve * G.  
\]
Note that, by Eq. \eqref{M2} and by the definition of $\rho^\ve$, one has that $\ve (\rho^\ve, \rho^\ve * G) /\ell =1 $,  hence $(\rho^\ve,\phi^\ve) = 0 $.  Such decomposition of a  function $\psi$ allows us  to rewrite the energy functional $E^\ve$ in a more convenient form in terms of $q^\ve$ and $\phi^\ve$. We first note that by the definition of $q^\ve$, for any $\psi\in H^1(\RE^3)$,  one has  
\[
E^\ve [\psi] = \| \nabla \psi\|^2- \frac{\ell}{\ve}|q^\ve|^2 +\frac{\gamma}{\mu+1}|q^\ve|^{2+2\mu}. 
\]
Next we note 
\begin{align*} \| \nabla \psi\|^2 & =  \| \nabla \phi^\ve + q^\ve \nabla \rho^\ve *G\|^2  \\
& =  \| \nabla \phi^\ve\|^2 + |q^\ve|^2 \|\nabla \rho^\ve *G\|^2 + 2 \Re q^\ve (\nabla \phi^\ve, \nabla \rho^\ve * G ) \\ 
& =  \| \nabla \phi^\ve\|^2 + \frac{\ell}{\ve} |q^\ve|^2 ,
 \end{align*}
where we used 
\[  \|\nabla \rho^\ve *G\|^2  =  (\nabla \rho^\ve *G,  \nabla \rho^\ve *G)  =   ( \rho^\ve ,   \rho^\ve *G) = \frac{\ell}{\ve},  \]
and 
\[ (\nabla \phi^\ve, \nabla \rho^\ve * G ) = ( \phi^\ve,  \rho^\ve  ) =0 .   \]
We conclude that the energy functional  $E^\ve[\psi]$ can be rewritten as 
\begin{equation} \label{energy2}E^\ve [\psi] = \| \nabla \phi^\ve\|^2 +\frac{\gamma}{1+\mu}|q^\ve|^{2+2\mu}. \end{equation}

Now we prove a simple proposition on the convergence of the initial data needed in the proof of the main Theorem \ref{t:main} (see Proposition \ref{l:psiconv}).
\begin{proposition}\label{p:initdata}
Let $\psi_0=\phi_0+q_0G\in \DD$ and define  $\psi_0^\ve=\phi_0+q_0G*\rho^{\ve}$. Then there exists a positive constant $C$ such that for all   $\ve>0$
\begin{equation}\label{5.4a}
\|\psi_0 - \psi_0^\ve\| \leq C \ve^{1/2}.
\end{equation}
Moreover $\psi_0^\ve\in H^2(\RE^3)$.     
\end{proposition}
\begin{proof}
We note that 
\[\|\psi_0 - \psi_0^\ve\| = |q_0| \|G-\rho^\ve*G\|.\]
In Fourier transform, one has that,
\begin{equation}\label{5.4c}
\|G-\rho^\ve*G\|^2 =\ve  \int_{\RE^3}\dk \frac{(\hat\rho(k) - \hat\rho(0))^2}{k^4} \leq \ve C \left( \int_{|\kk|<1} \dk  +\int_{|\kk|\geq1} \dk \frac{1}{k^4}\right) \leq C\ve ,
\end{equation}
where we used  bound \eqref{working}.  

The fact that $\psi_0^\ve\in L^2(\RE^3) $ is an immediate consequence of \eqref{5.4a}. Moreover $\Delta\phi_0\in L^2(\RE^3)$ and $\|\Delta \rho^\ve * G\|$ is finite for any $\ve>0$, hence $\Delta\psi^{\ve}_0\in L^2(\RE^3)$.  
 \end{proof}

Next we prove several useful bounds on $q^\ve(t)$ defined as in \eqref{apprcharge}.  
Consider the solution $\psi^{\ve}(t)$ of  problem \eqref{weakap3d} corresponding to initial datum $\psi^{\ve}_0=\phi_0+q_0G*\rho^{\ve}$. We set 
\begin{equation}
\label{5.4b}
\phi^\ve(t) = \psi^\ve(t) - q^\ve(t) \rho^\ve* G.
\end{equation}
The following result gives a priori bounds on $\phi^\ve(t)$ and $q^\ve(t)$ as deduced by the conserved energy functional rewritten in the form \eqref{energy2} and evaluated on $\psi^{\ve}(t).$
\begin{proposition}\label{l:5.2}
Let $\psi_0=\phi_0+q_0G\in \DD$, define  $\psi_0^\ve=\phi_0+q_0G*\rho^{\ve}$ and fix $T>0$. Then if $\gamma\geq 0$ and $\mu\geq 0$ or $\gamma<0$ and $0\leq\mu<1$, there exist two positive constants  $\ve_0$ and $C$ not depending on $\ve$ such that, for all  $0<\ve<\ve_0$, 
\begin{equation}\label{apriori0}
 \sup_{t\in [0,T]}\|\nabla\phi^\ve(t)\| \leq C
\end{equation}
and
\begin{equation}\label{apriori1}
 \|q^{\ve}\|_{L^\infty(0,T)} \leq C.
\end{equation}
\end{proposition}
\begin{proof}
We note the following inequalities 
\begin{equation} \label{2}
\ve|(\rho^{\ve},\psi)| \leq C\ve^{-1/2}\|\psi\|\;; \qquad  \ve| (\rho^{\ve},\psi)| \leq C \ve^{1/2}  \|\nabla\psi\|\ .
\end{equation} 
The first one holds true for any $\psi\in L^2(\RE^3)$ and the second one for $\psi\in \dot H^1(\RE^3). $ Both of them follow directly from the Cauchy-Schwarz inequality. In the first one we use $\|\rho^\ve\| \leq C \ve^{-3/2}$. In the second one we use $|(\rho^\ve,\psi) |\leq \||\cdot|^{-1} \hat \rho^\ve\|\, \||\cdot|\hat \psi\|$ and $\||\cdot|^{-1} \hat \rho^\ve\| = \sqrt{\ell/\ve}$, by \eqref{M2}. 

By the conservation of the energy and by identity \eqref{energy2} we have that 
\[E^\ve[\psi^\ve(t)] = \|\nabla \phi^\ve(t)\|^2 + \frac{\gamma}{1+\mu} |q^\ve(t)|^{2+2\mu} = \|\nabla \phi^\ve_0\|^2 + \frac{\gamma}{1+\mu} |q^\ve_0|^{2+2\mu} = E^\ve[\psi_0^\ve]  \]
where, according to  \eqref{qve} and \eqref{5.4b},  $q_0^\ve$ and $\phi_0^\ve$ are given by
\begin{equation}\label{5.7a}
q_0^\ve : = \frac{\ve}{\ell} (\rho^\ve,\psi_0^\ve)
\end{equation}
and 
\begin{equation}\label{phi0ve}
\phi_0^\ve : = \psi_0^\ve - q_0^\ve \rho^\ve*G.
\end{equation}
In Eq. \eqref{5.7a}, we use Eq. \eqref{psi0ve}  and the fact that $\ve(\rho^\ve, \rho^\ve *G)/\ell = 1$, rearranging the terms, we end up with  
\[q_0^\ve - q_0 = \frac{\ve}{\ell} (\rho^\ve,\phi_0) .\]
By the second inequality in \eqref{2} it follows that 
\begin{equation} \label{machine}|q_0^\ve - q_0| \leq  C  \ve^{1/2} \|\nabla\phi_0\|.\end{equation} 
The latter bound, together with Eqs. \eqref{phi0ve}  gives 
\begin{equation} \label{machine2}
\|\nabla \phi_0^\ve\| = \|\nabla \phi_0 + (q_0 -   q_0^\ve) \nabla \rho^\ve*G   \| \leq C \|\nabla \phi_0\|,
\end{equation}
where we used the identity  $\|\nabla \rho^\ve*G   \| =\sqrt{ \ell/\ve}$.\\
Next we distinguish the three cases: $\gamma>0$, $\gamma <0$, and $\gamma=0$. 

For $\gamma > 0$ the energy functional $E^\ve$ is nonnegative. This piece of information, together with  \eqref{machine} and \eqref{machine2}, and the inequality $|q_0^{\ve}|^{2+2\mu} \leq C(|q_0^{\ve}-q_0|^{2+2\mu} +|q_0|^{2+2\mu} )$, give the following chain of inequalities 
\begin{equation}\label{burning}
0<E^\ve[\psi_0^\ve] \leq C (\|\nabla\phi_0\|^2+\ve^{1+\mu}\|\nabla\phi_0\|^{2+2\mu} + |q_0|^{2+2\mu})\leq  C_0.
 \end{equation}
By the conservation of the energy, this implies that $0 < E^\ve[\psi^\ve(t)] \leq C_0$. Since $E^\ve[\psi^\ve(t)]$ is the sum of two positive terms, both of them must satisfy the same bound, and this concludes the proof of  bounds \eqref{apriori0} and \eqref{apriori1} for $\gamma>0$. \\

For $\gamma<0$  the proof is more involved because the sign of the energy is not definite. We start by noting that, by the same argument as above, one has a bound similar to \eqref{burning} for the absolute value of the energy, precisely 
\[|E^\ve[\psi_0^\ve]| \leq C_0.\]
 By the conservation of the energy this gives  
\begin{equation}\label{stone}|E^\ve[\psi^\ve(t)]| \leq C_0.\end{equation}
We claim that  for any $0\leq \mu<1$  the following bound holds true 
\begin{equation}\label{mainbound} |q^\ve(t)| \leq \widetilde C \|\nabla \phi^\ve(t)\|^{1/2}.\end{equation}
By using this bound and \eqref{stone}, we find the chain of  inequalities 
\begin{equation} \label{5.14a}
C_0\geq E^\ve[\psi^\ve(t)] \geq \|\nabla \phi^\ve(t)\|^2 - \frac{|\gamma|\widetilde C}{1+\mu}  \|\nabla \phi^\ve(t)\|^{1+\mu}. 
\end{equation}
Noticing that if  $x^2 -\tilde cx^{1+\mu} \leq c_0$, with $\tilde c,c_0>0$, $x>0$, and $0\leq\mu<1$ then it must be  $x\leq c$, we conclude that Eq. \eqref{5.14a} implies  $\|\nabla \phi^\ve(t)\|\leq C$, hence $|q^\ve(t)|\leq C$,  for all $0\leq\mu<1$. \\
We are left to prove the claim in  \eqref{mainbound}. To this aim we start by noting that  by the first inequality in \eqref{2} and by the conservation of the mass 
\[|q^\ve(t)| \leq \frac{\ve}{\ell} \|\rho^\ve\| \,\|\psi^\ve_0\| \leq C \ve^{-1/2},\]
where we used the fact that by Prop. \ref{p:initdata} $\|\psi_0^\ve\| \leq C$.  By using this rough bound on $|q^\ve(t)|$ together with  bound \eqref{stone}, we obtain the chain of inequalities  
\[
C_0 \geq E^\ve[\psi^\ve(t)]\geq  \|\nabla \phi^\ve(t)\|^2 - \frac{|\ga|C}{\mu+1} \ve^{-(1+\mu)}
\]
which in turn implies that for $0<\ve<1$
\[  \|\nabla \phi^\ve(t)\|^2\leq C_0+ \frac{|\ga|C}{\mu+1} \ve^{-(1+\mu)} \leq C(1+ \ve^{-(1+\mu)}) \leq  C\ve^{-(1+\mu)}.\]
Hence we have the following bound on $ \|\nabla \phi^\ve(t)\|$
 \begin{equation}\label{boundapriori}
\|\nabla \phi^\ve(t)\| \leq C \ve^{-(1+\mu)/2}.\end{equation}
Now we are ready to prove   bound \eqref{mainbound}. Take $\rho^{\ve_1} $, where $\ve_1$ is a parameter which will be fixed later on.  \\
By the first inequality in  \eqref{2}, together with the  conservation of the mass and Prop. \ref{p:initdata}, we get  
\[\ve_1|(\rho^{\ve_1},\psi^\ve(t))| \leq C \ve_1^{-1/2}. \] 
On the other hand the second inequality in  \eqref{2} gives 
\[
\ve_1|(\rho^{\ve_1},\phi^\ve(t))| \leq C \ve_1^{1/2} \|\nabla\phi^\ve(t)\|. \]
We can use these two bounds together with the rough bound \eqref{boundapriori}  to prove that  $|q^\ve(t)|$ is bounded uniformly in $\ve$. We take the scalar product of  identity \eqref{5.4b} for $\rho^{\ve_1}$ and rearrange the terms to obtain  
\begin{equation}\label{4}
|q^\ve(t)| \leq \frac{ |(\rho^{\ve_1},\psi^\ve(t))|  + | (\rho^{\ve_1},\phi^\ve(t))| }{|(\rho^{\ve_1},\rho^\ve*G)|}
\leq C \frac{ \ve_1^{-1/2}  + \ve_1^{1/2}\|\nabla \phi^\ve(t)\|}{\ve_1|(\rho^{\ve_1},\rho^\ve*G)|} .
\end{equation}
The latter bound gives \eqref{mainbound} if we set $\ve_1 = \|\nabla\phi^\ve(t)\|^{-1}$ and we can guarantee that with this choice of $\ve_1$ the denominator at the r.h.s. does not go to zero as $\ve$ goes to zero.  Let us check that this is indeed the case. We note that 
\[\begin{aligned}
\ve_1(\rho^{\ve_1},\rho^\ve*G) = & \frac{\ve_1}{\ve} \int_{\RE^3}\dk  \frac{\hat \rho(\ve_1 k/\ve) \hat \rho(k)  }{|\kk|^2}  \\  
   = & A_1 + A_2^{\ve_1/\ve}
\end{aligned}\]
with 
\[
A_1 =  \hat \rho(0) \int_{\RE^3}\dk  \frac{ \hat \rho(k)  }{|\kk|^2} = \int_{\RE^3} \dx \frac{\rho(|\xx|) }{4\pi |\xx|} ,
\]
and 
\[
A_2^{\ve_1/\ve} = \frac{\ve_1}{\ve} \int_{\RE^3}\dk  \frac{\hat \rho(\ve_1 k/\ve)( \hat \rho(k)- \hat \rho(0))  }{|\kk|^2 }  .
\]
Then we note that $A_1>0$, because $\rho$ is non-negative, and that  
\[\begin{aligned}
|A_2^{\ve_1/\ve}| \leq & \frac{\ve_1}{\ve}  \left(\int_{\RE^3}\dk (\hat \rho(\ve_1 k/\ve))^2\right)^{1/2} \left(\int_{\RE^3}\dk  \frac{( \hat \rho(k)- \hat \rho(0))^2  }{|\kk|^4 }\right)^{1/2} \leq  C  \sqrt{\frac{\ve}{\ve_1}} , 
\end{aligned}\]
where we used the same argument as in \eqref{5.4c}. Recalling that $\ve_1 = \|\nabla \phi^\ve(t)\|^{-1}$ and the rough bound  \eqref{boundapriori}, we obtain  
\[
\frac{\ve}{\ve_1} = \ve \|\nabla \phi^\ve(t)\| \leq C \ve^{(1-\mu)/2}.
\]
Which in turn implies that $|A_2^{\ve_1/\ve}| \leq C \ve^{(1-\mu)/4}$, and that, for any $0\leq\mu<1$ and $\ve$ small enough, 
\begin{equation}\label{lower}
|\ve_1(\rho^{\ve_1},\rho^\ve*G)| \geq ||A_1| -|A_2^{\ve_1/\ve}|| \geq |A_1|/2.
\end{equation}
By using this lower bound in \eqref{4}, together with the fact that we set $\ve_1 = \|\nabla \phi^\ve(t)\|^{-1}$, we obtain bound \eqref{mainbound}, and conclude the proof of the lemma for $\gamma<0$. \\ 

It remains to show that for $\gamma=0$  bounds \eqref{apriori0} and  \eqref{apriori1} hold for any $\mu\geq 0$. This is proven by noticing that in  this case  the conservation of the energy  immediately implies  the bound $\|\nabla\phi^\ve(t)\| \leq C$. To prove \eqref{apriori1}, note that  inequality \eqref{4} is still satisfied, because it is a direct  consequence of the decomposition in Eq. \eqref{5.4b}. Taking again  $\ve_1 = \|\nabla\phi^\ve(t)\|^{-1}$ one has that the lower bound \eqref{lower} holds true for any $\mu\geq 0$, because $\ve/\ve_1 \leq C\ve $. Hence  the bound \eqref{mainbound} holds true for any $\mu\geq 0$, thus implying  $|q^\ve(t)| \leq \|\nabla \phi^\ve(t)\| \leq C$.
\end{proof}

\begin{remark}
Notice that the above uniform bounds hold true only for the values of $\mu$ and $\gamma$ for which the limit problem admits a global solution. The scaled problem is well posed for a much wider range of parameters. The restricted range of parameters in the limit problem   is related to the lack of uniform bounds. An analogous phenomenon occurs for the one dimensional case, see \cite{CFNT1}, Remark 2.1.
\end{remark}
 
\section{A priori estimates on the regularized charge\label{s:5}}

This section  lays the groundwork for the proof of the main theorem in the next section. First we cast  Eq. \eqref{eqapprchargeeq} in an equivalent form, then we use it to derive a priori estimates on the derivative and half derivative of the approximating charge $q^\ve(t)$ (see  Prop.  \ref{l:5.4}).

\begin{proposition}\label{p:5.1}
If the function $q^\ve(t)$ satisfies  Eq. \eqref{eqapprchargeeq}, it also satisfies:
\begin{equation}\label{Iqve}
I^{1/2} q^\ve(t) = - 4\pi \sqrt{\pi i} \ga \int_0^t \ds  |q^\ve(s)|^{2\mu} q^\ve(s)  + 4\pi \sqrt{\pi i}  \int_0^t \dd s \, (\rho^\ve,U(s)*\psi^\ve_0)  + \sum_{j=1}^3Y_j^\ve(t), 
\end{equation}
with 
\begin{equation}
\label{Yve1}
Y_1^\ve(t) =    
   -(4\pi)^2  \sqrt{\pi i}   \int_0^t \dd \tau q^\ve(\tau)    \int_0^\infty \dd k\, ( (\hat \rho(\ve k))^2 -(\hat\rho(0))^{2})\, e^{-ik^2(t-\tau)},  \end{equation}
\begin{equation}\label{Yve2}
Y_2^\ve(t) =    (4\pi)^2  \sqrt{\pi i}   \gamma \frac{\ve}{\ell}  \int_0^t \dd \tau |q^\ve(\tau)|^{2\mu} q^\ve(\tau)     \int_0^\infty \dd k\, ( (\hat \rho(\ve k))^2 -(\hat\rho(0))^{2})\, \, e^{-ik^2(t-\tau)},
\end{equation}
and
\begin{equation}
\label{Yve3}
Y_3^\ve(t) =    \gamma \frac{\ve}{\ell}  \int_0^t \dd \tau \frac{|q^\ve(\tau)|^{2\mu} q^\ve(\tau)}{\sqrt{t- \tau}}.
\end{equation}
\end{proposition}
\begin{proof}
By integrating Eq. \eqref{eqapprchargeeq} and rewriting $(\rho^\ve,U(t)\rho^\ve)$ in  Fourier transform we obtain the identity
\begin{equation} \begin{aligned}\label{down}
\frac{\ell}{\ve} \int_0^t \dd s \,  q^\ve(s )  = & \int_0^t \dd s \, (\rho^\ve,U(s)\psi^\ve_0)  \\ 
& + 4\pi  i  \int_0^t \dd s \, \int_0^s \dd \tau  \int_0^\infty \dd k\, k^2 (\hat\rho(\ve k))^2 e^{-ik^2(s-\tau)} \left(1 - \gamma\frac{\ve}{\ell} |q^\ve(\tau)|^{2\mu}\right) q^\ve(\tau).
\end{aligned}\end{equation}
In the latter term, we exchange the order of integration and perform the integral with respect to the variable $s$. We obtain
\begin{align*}
 &4\pi  i  \int_0^t \dd s  \int_0^s \dd \tau  \int_0^\infty \dd k\, k^2 (\hat \rho(\ve k))^2 e^{-ik^2(s-\tau)} \left(1 - \gamma\frac{\ve}{\ell} |q^\ve(\tau)|^{2\mu}\right) q^\ve(\tau) \\ 
=  &  4\pi \int_0^t \dd \tau \left(1 - \gamma\frac{\ve}{\ell} |q^\ve(\tau)|^{2\mu}\right) q^\ve(\tau)    \int_0^\infty \dd k\,  (\hat \rho(\ve k))^2(1- e^{-ik^2(t-\tau)})  \nonumber\\ 
=  &  \frac{\ell}{\ve} \int_0^t \dd \tau  q^\ve(\tau)  -   \gamma \int_0^t \dd \tau |q^\ve(\tau)|^{2\mu} q^\ve(\tau)  - 4\pi   \int_0^t \dd \tau q^\ve(\tau)    \int_0^\infty \dd k\,  (\hat \rho(\ve k))^2 \, e^{-ik^2(t-\tau)}  \\ 
  & +4\pi   \gamma\frac{\ve}{\ell}  \int_0^t \dd \tau |q^\ve(\tau)|^{2\mu} q^\ve(\tau)     \int_0^\infty \dd k\,  (\hat \rho(\ve k))^2 \, e^{-ik^2(t-\tau)}  \nonumber\\ 
=  &  \frac{\ell}{\ve} \int_0^t \dd \tau  q^\ve(\tau)  -  \gamma \int_0^t \dd \tau |q^\ve(\tau)|^{2\mu} q^\ve(\tau)     -  \frac{1}{4\pi \sqrt{\pi i}}\int_0^t \dd \tau \frac{ q^\ve(\tau)}{\sqrt{t-\tau}} +  \frac{1}{4\pi \sqrt{\pi i}}\sum_{j=1}^3 Y_j^\ve(t).
\end{align*}
Here we used  identity \eqref{M2}, the Fresnel integral \eqref{identities1}, and $\hat\rho(0)=(2\pi)^{-3/2}$.  Going back to  Eq. \eqref{down} and rearranging the terms we obtain  Eq.  \eqref{Iqve}. 
\end{proof}

In the following Proposition we prove the bounds on the derivative and half-derivative of $q^\ve$. 
\begin{proposition}\label{l:5.4}
Let $\psi_0=\phi_0+q_0G \in \DD$ and $\psi_0^\ve=\phi_0+q_0\rho^\ve*G$ and fix $T>0$. Then if $\gamma\geq 0$ and $\mu\geq 0$ or $\gamma<0$ and $0\leq\mu<1$, there exist two positive constants  $\ve_0$ and $C$  such that, for all  $0<\ve<\ve_0$, 
\begin{equation}
\label{apriori2}
\|\dot q^\ve\|_{L^\infty(0,T)} \leq C\ve^{-3/2},
\end{equation}
and 
\begin{equation}
\label{apriori3}
\|D^{1/2} q^\ve\|_{L^1(0,T)} \leq C\ve^{-(1/2+\de)}  
\end{equation}
for all $\de>0$ with the constant $C$ in \eqref{apriori3} depending on $\de$.
\end{proposition}
\begin{proof}
We note that by Prop. \ref{p:initdata},  $\psi_0^\ve\in H^2(\RE^3)$ for all $\ve>0$. Hence $\psi^\ve(t)$ satisfies Eq. \eqref{appeq} that can be rewritten as 
\[i\pd{}{t}\psi^\ve(t)= -\Delta \psi^\ve(t) -q^\ve(t)  \rho^\ve + \ga \frac{\ve}{\ell} |q^\ve(t)|^{2\mu} q^\ve(t) \rho^\ve. 
\]
Taking the scalar product  for $\ve \rho^\ve/\ell $ we get  
\[\begin{aligned}\dot q^\ve(t) =&  \frac{\ve}{\ell}\left(\rho^\ve,\pd{}{t} \psi^\ve(t)\right) \\
 =& -i \frac{\ve}{\ell}\left((\rho^\ve,-\Delta \psi^\ve(t)) - q^\ve(t) (\rho^\ve,\rho^\ve)+ \frac{\ga \ve}{\ell} |q^\ve(t)|^{2\mu} q^\ve(t)  (\rho^\ve,\rho^\ve)\right)  . \end{aligned}\]
 We rewrite  $\psi^\ve(t)$ as in Eq. \eqref{5.4b} and note that 
 \[ (\rho^\ve,-\Delta \psi^\ve(t)) = (\rho^\ve,-\Delta \phi^\ve(t)) + q^\ve(t) (\rho^\ve, -\Delta \rho^\ve * G)
 =  (\nabla\rho^\ve,\nabla\phi^\ve(t)) + q^\ve(t) (\rho^\ve, \rho^\ve ),\] where we used $(\rho^\ve, -\Delta \rho^\ve * G)
=(\rho^\ve, \rho^\ve )$.
 Hence 
 \[
 \dot q^\ve(t) = -i \frac{\ve}{\ell}\left((\nabla\rho^\ve,\nabla\phi^\ve(t)) + \frac{\ga \ve}{\ell} |q^\ve(t)|^{2\mu} q^\ve(t)  (\rho^\ve,\rho^\ve)\right) . 
 \]
Bound \eqref{apriori2} then follows by noting that 
\[|(\nabla\rho^\ve,\nabla\phi^\ve(t))| \leq \|\nabla\rho^\ve\|\|\nabla\phi^\ve(t)\| = \ve^{-5/2} \|\nabla\rho\|\|\nabla\phi^\ve(t)\| \leq C\ve^{-5/2},\] by Proposition \ref{l:5.2}, and  that $\|\rho^\ve\|^2 \leq C \ve^{-3}$. \\

To get bound \eqref{apriori3} we take the derivative of Eq. \eqref{Iqve} and obtain
\begin{equation}\label{problema}
D^{1/2} q^\ve(t) =  
 - 4\pi \sqrt{\pi i}\gamma   |q^\ve(t)|^{2\mu} q^\ve(t)   +4\pi \sqrt{\pi i}  (\rho^\ve,U(t)\psi^\ve_0)  + \sum_{j=1}^3\frac{d}{dt}Y_j^\ve(t).
 \end{equation}
By \eqref{apriori1} the first term at the r.h.s. satisfies $\| |q^\ve|^{2\mu} q^\ve\|_{L^1(0,T)} \leq C$. 

For  the second term at the r.h.s of Eq. \eqref{problema} we use the decomposition \eqref{psi0ve}, which gives 
\[
| (\rho^\ve,U(t)\psi^\ve_0) | \leq | (\rho^\ve,U(t)\phi_0) |+|q_0|\,| (\rho^\ve,U(t) (\rho^\ve *G)) |.
\]
By the unitarity of $U(t)$ and the second inequality in \eqref{2}, we have that 
\[ | (\rho^\ve,U(t)\phi_0) | \leq C\ve^{-1/2} \|\nabla\phi_0\| \leq C\ve^{-1/2}  ,\] recall that $\DD \subset \EE$ by Rem. \ref{r:EEDD}. 
Moreover one has 
\begin{equation}\label{5.35}
| (\rho^\ve,U(t) (\rho^\ve *G)) |  \leq \frac{C}{\sqrt t}. 
\end{equation}
The proof of this claim is postponed at the end.
From \eqref{5.35} we have 
\[\|(\rho^\ve,U(\cdot) \psi_0^\ve) \|_{L^1(0,T)} \leq \|(\rho^\ve,U(\cdot)\phi_0) \|_{L^1(0,T)}  + |q_0| \|(\rho^\ve,U(\cdot) (\rho^\ve *G)) \|_{L^1(0,T)} \leq  C\ve^{-1/2}.\] 

To conclude the proof of  bound \eqref{apriori3}  we need to  show that 
\begin{equation} \label{treno2}
\left\| \frac{d}{dt}Y_j^\ve(t)  \right\|_{L^1(0,T)}  \leq \frac{C}{\ve^{\frac12 + \de}} 
\end{equation}
 for all $\de>0$ and $j=1,2,3$. Since $\ve$ is small, it is enough to prove the inequality for $0<\de\leq1$.
We start the analysis with the term  $Y_1^\ve $ and prove that
\begin{equation} \label{treno2.1}
\left\| \frac{d}{dt}Y_1^\ve(t)  \right\|_{L^1(0,T)}  \leq \frac{C}{\ve^{\frac12 + \de}} .
\end{equation}
Taking the derivative in \eqref{Yve1} we get 
\[
\begin{aligned}
\frac{d}{dt}Y_1^\ve(t) =  &   
   -(4\pi)^2  \sqrt{\pi i}  \int_0^t \dd \tau \dot q^\ve(t-\tau)    \int_0^\infty \dd k\, ( (\hat \rho(\ve k))^2 -(\hat\rho(0))^{2})\, e^{-ik^2 \tau}    
  \\  & - (4\pi)^2  \sqrt{\pi i}   q^\ve(0)    \int_0^\infty \dd k\, ( (\hat \rho(\ve k))^2 -(\hat\rho(0))^{2})\, e^{-ik^2t}   .
\end{aligned} 
\]
We use the further bound (see below for the proof)
\begin{equation}\label{5.38}
\left|\int_0^\infty \dd k\, ( (\hat \rho(\ve k))^2 -(\hat\rho(0))^{2})\, e^{-ik^2t}   \right| \leq C \frac{\ve^\eta}{t^{(1+\eta)/2}} \qquad \forall \, 0\leq\eta < 1,
\end{equation}
which, together with bounds \eqref{apriori1} and \eqref{apriori2}, give 
\[
\left|\frac{d}{dt}Y_1^\ve(t) \right| \leq 
C\left( 
\frac{1}{\ve^{3/2-\eta}}  \int_0^t \dd \tau  \frac{1}{\tau^{(1+\eta)/2}} +   \frac{\ve^{\eta}}{t^{(1+\eta)/2}}   \right).
\]
Bound \eqref{treno2.1} now follows by taking $\eta = 1-\de$ with $0<\de\leq1$ and using the fact that $t^{-(1+\eta)/2}$ is integrable in $[0,T]$.  \\

The term $Y_2^\ve$ can be studied in the same way as $Y_1^\ve$. It is enough to change $q^\ve$ in $|q^\ve|^{2\mu}q^\ve$, and use the bound 
\begin{equation}\label{andatura}
\left| \frac{d}{dt} |q^\ve(t)|^{2\mu} q^\ve(t) \right| \leq \frac{C}{\ve^{3/2}},
\end{equation}
 taking into account the factor $\ve$, one can see that 
\[
\left\| \frac{d}{dt}Y_2^\ve(t)  \right\|_{L^1(0,T)}  \leq C\ve^{\frac12-\de} .\]
To bound the  term $Y_3^\ve$ we take the derivative and obtain 
\[
\frac{d}{dt}Y_3^\ve(t) =    \gamma\frac{\ve}{\ell}  \int_0^t \dd \tau \frac{1}{\sqrt{ \tau}}\frac{d}{dt} |q^\ve(t-\tau)|^{2\mu} q^\ve(t-\tau)
+ \gamma\frac{\ve}{\ell}  \frac{|q^\ve(0)|^{2\mu} q^\ve(0)}{\sqrt{t}}.
\]
By using again   bound \eqref{andatura} and since   $1/\sqrt t$ is integrable one has 
\[
\left\| \frac{d}{dt}Y_3^\ve(t)  \right\|_{L^1(0,T)}  \leq \frac{C}{\ve^{1/2}}.
\]
And this concludes the proof of \eqref{treno2}. 

To close the proof it remains to get  bounds \eqref{5.35} and \eqref{5.38}. We start with the second one, noticing that by Eq. \eqref{working} one has $|(\hat \rho(k))^2 -(\hat\rho(0))^{2}|=|\hat \rho(k) -\hat\rho(0)||\hat \rho(k)+\hat\rho(0)| \leq C k^{\eta}$ for all $0\leq\eta \leq 2$. We change variables in the integral in \eqref{5.38},  
\[
\left|\int_0^\infty \dd k\, ( (\hat \rho(\ve k))^2 -(\hat\rho(0))^{2})\, e^{-ik^2t}   \right| = 
\frac{1}{\sqrt t}\left|\int_0^\infty \dd k\, ( (\hat \rho(\ve k\sqrt t))^2 -(\hat\rho(0))^{2})\, e^{-ik^2}   \right|,
\]
and distinguish two cases. If $\ve/\sqrt t>1$ then 
\[\left|\int_0^\infty \dd k\, ( (\hat \rho(\ve k\sqrt t))^2 -(\hat\rho(0))^{2})\, e^{-ik^2}   \right| 
\leq \frac{\sqrt t}{\ve}\int_0^\infty \dd k\,  (\hat \rho( k))^2 + (\hat\rho(0))^{2} \left|\int_0^\infty \dd ke^{-ik^2}   \right| \leq C \left(\frac{\ve}{\sqrt t}\right)^\eta  \]
for all $\eta\geq 0$. If $\ve/\sqrt t\leq 1$ we  use the identity $e^{-ik^2} = -\frac{1}{2ik} \frac{d}{dk} e^{-ik^2}$, and integrate by parts to obtain 
 \[\begin{aligned}
&  \left|  \int_0^\infty dk ( (\hat \rho(\ve k/\sqrt t ))^2 - (\hat \rho(0))^2) e^{-ik^2 } \right| =    \left|  \int_0^\infty dk  \frac{ ( (\hat \rho(\ve k/\sqrt t ))^2 - (\hat \rho(0))^2) }{2k} \frac{d}{dk} e^{-ik^2} \right| \\ 
   \leq &
C\left( \frac{\ve}{\sqrt t}  \int_0^\infty dk \frac{|\hat \rho(\ve k/\sqrt t )| \,   |\hat \rho'(\ve k/\sqrt t )|}{k} + 
    \int_0^\infty dk   \frac{|(\hat \rho(\ve k/\sqrt t ))^2- (\hat \rho(0))^2|}{k^2} \right) \\ 
 \leq  & C   \frac{\ve}{\sqrt t}  \left(   \int_0^\infty dk \frac{|\hat \rho( k )| \,   |\hat \rho'( k )|}{k} + 
   \int_0^\infty dk   \frac{|(\hat \rho(k ))^2- (\hat \rho(0))^2|}{k^2} \right)
\leq   C \left(\frac{\ve}{\sqrt t}\right)^\eta  , \end{aligned}\]
where we took $0\leq \eta \leq 1$.  And this concludes the proof of  bound \eqref{5.38}.
To prove the claim in \eqref{5.35} we note first that 
\[| (\rho^\ve,U(t) (\rho^\ve *G)  ) |  = 4\pi \left|  \int_0^\infty dk (\hat \rho(\ve k ))^2 e^{-ik^2 t} \right| \leq 
C\left|  \int_0^\infty dk\big( (\hat \rho(\ve k ))^2 - (\hat \rho(0 ))^2 \big)e^{-ik^2 t} \right| + \frac{C}{\sqrt t}
\]
and then proceed by using \eqref{5.38}.
\end{proof}

\section{Convergence to the limit dynamics \label{s:6}}

In this section we prove Th. 1, which states the convergence of $\psi^\ve(t)$ to $\psi(t)$ in the $L^2$-norm for initial data in the operator  domain.  The idea of the proof is to show that the regularized $q^\ve(t)$ converges to $q(t)$ in a suitable topology. We then use this result to prove the convergence of $\psi^\ve(t)$ to $\psi(t)$. 
To this end the following proposition has a central role. There we show that for initial data in $\DD$ the $L^2$-convergence  of $\psi^\ve$ to $\psi$ can be reduced to the convergence of the corresponding initial data in $L^2(\RE^3)$ and of $I^{1/2}q^\ve$ to $I^{1/2}q$ in $L^{\infty}(\RE)\ .$
\begin{proposition}\label{l:psiconv}
Let $\psi_0=\phi_0+q_0G \in \DD$ and $\psi_0^\ve=\phi_0+q_0\rho^\ve*G$  and fix $T>0$. Then if $\gamma\geq 0$ and $\mu\geq 0$ or $\gamma<0$ and $0\leq\mu<1$, there exist two positive constants  $\ve_0$ and $C$  such that, for all  $0<\ve<\ve_0$, 
\[
\sup_{t\in[0,T]} \|\psi^\ve(t) - \psi(t)\| \leq C \left( \|\psi_0^\ve-\psi_0\|^{1/2}+\|I^{1/2}(q^\ve -q)\|_{L^\infty(0,T)}^{1/2} + \ve^{1/4}\right).
\]
\end{proposition}
\begin{proof}We shall prove that 
\[
\sup_{t\in[0,T]}\|\psi^\ve(t) - \psi(t)\|^2 \leq C \left( \|\psi_0^\ve - \psi_0\|+ \|I^{1/2}(q^\ve   - q) \|_{L^\infty(0,T)}
  + \ve^{1/2}\right).
\]
We start by noting that, by the mass conservation, 
\[\|\psi^\ve(t) - \psi(t)\|^2 \leq \big| \|\psi_0^\ve \|^2 - \|\psi_0 \|^2\big| + 2 |(\psi(t) , \psi^\ve(t) - \psi(t))|.  \] 
For the first term at the r.h.s. we use the trivial bound $ \big| \|\psi_0^\ve \|^2 - \|\psi_0 \|^2\big|  \leq   \|\psi_0^\ve - \psi_0 \|   \big( \|\psi_0^\ve \| + \|\psi_0 \|\big) $. 

We are left to prove that 
\begin{equation}\label{orlando}
|(\psi(t) , \psi^\ve(t) - \psi(t))| \leq  C \left(  \|\psi_0^\ve - \psi_0\|+  \|I^{1/2}(q^\ve   - q) \|_{L^\infty(0,T)} + \ve^{1/2}\right),
\end{equation}for all $t\in[0,T]$. 
We start by using  identity  \eqref{solinhom} for  the solution $\psi(t)$ of the limit nonlinear flow, and  identity \eqref{solveps} for the solution $\psi^\ve(t)$ of the scaled nonlinear flow and we get 
\begin{equation}\begin{aligned}\label{black}
|(\psi(t) , \psi^\ve(t) - \psi(t))| \leq & |(\psi(t) , U(t)(\psi^\ve_0 - \psi_0))| 
+ \left| \int_0^t ds (\psi(t), q^\ve(s) U(t-s) \rho^\ve  - q(s)U(t-s,\cdot) ) \right|  \\ &+
|\gamma| \frac{\ve}{\ell}\left| \int_0^t ds  |q^\ve(s)|^{2\mu}q^\ve(s) (\psi(t), U(t-s) \rho^\ve) \right|. 
\end{aligned}\end{equation}
For the first term at the r.h.s of Eq. \eqref{black}, we use the Cauchy-Schwarz inequality, the  unitarity of the free evolution group and  the conservation of mass to get the bound 
\[ |(\psi(t) , U(t) (\psi^\ve_0 - \psi_0))| \leq  \|\psi_0\|\|\psi^\ve_0 - \psi_0\|. \] 
Hence, to conclude the proof of \eqref{orlando}, we show that 
\begin{equation} \label{black2}
\left| \int_0^t ds (\psi(t), q^\ve(s) U(t-s) \rho^\ve  - q(s)U(t-s,\cdot) ) \right| \leq  C \left(  \|I^{1/2}(q^\ve   - q) \|_{L^\infty(0,T)}+\ve^{1/2} \right), 
 \end{equation}
 and that 
 \begin{equation}\label{black3}
 \left| \int_0^t ds  |q^\ve(s)|^{2\mu}q^\ve(s) (\psi(t), U(t-s) \rho^\ve) \right| \leq  C\ . 
 \end{equation}

We start by proving  inequality \eqref{black2}. By adding and subtracting 
\[q^\ve(s)(\psi(t),U(t-s,\cdot)) = q^\ve(s)(U(t-s) \bar\psi(t))(\oo)\] in the integral we obtain
\[
\left| \int_0^t \dd s (\psi(t), q^\ve(s) U(t-s) \rho^\ve  - q(s)U(t-s,\cdot ) ) \right|  
\leq A^\ve(t) + R^\ve(t)\]
with 
\[
A^\ve(t) =    \left| \int_0^t \dd s (q^\ve(s) - q(s))  (U(t-s)  \bar \psi(t))(\oo)  \right| 
\]
and 
\[
R^\ve(t) =  \left| \int_0^t \dd s \, q^\ve(s) (\psi(t),  \big (U(t-s) \rho^\ve  - U(t-s,\cdot)) \big) \right| .
\]
We decompose the function $\psi(t)$ as 
$\psi(t) = \phi(t) + q(t) G $
where $\phi(t)\in \dot H^2(\RE^3)$ and  $q(t) \in \CO$. For the term $A^\ve(t)$ we obtain 
\[
A^\ve(t) \leq C \left( \left|\int_0^t \dd  s I^{1/2}(q^\ve -q)(s) \frac{d}{dt} I^{1/2} f (t-s) \right|+ |q(t)| \, |I^{1/2}(q^\ve -q)(t)| \right) 
\]
where we set $f(t) = (U(t)*\bar\phi(t))(\oo)$ and used the identities 
\begin{equation}
\label{5.43a}
\int_0^t \dd s  (q^\ve(s)   - q(s)) f(t-s) = \frac{1}{\pi} \int_0^t \dd s \,  I^{1/2}(q^\ve   - q)(s)  \,\frac{d}{dt} I^{1/2} f(t-s),
\end{equation}
see Eq.  \eqref{this},  and 
\begin{equation}\label{ray}
(U(t)G)(\oo) = \frac{1}{2\pi^2}\int_0^\infty \dd k  \, e^{-i k^2t} = \frac{1}{4\pi} \frac{1}{\sqrt{\pi i t}}.
\end{equation}
Since $\phi(t)\in \dot H^2(\RE^3)$, by Prop. \ref{p:harvey}, one has that  $f\in \dot H^{3/4}(\RE)$. We set $f_0 = f(0)$ and  $h = f - f_0$, so that $h\in \dot H^{3/4}(\RE)$ and $h(0) = 0$. We proceed as in the proof of Lemma \ref{l:wpq}, and define $ h_s:[0,2T]\to \CO$ by symmetrizing $h$ w.r.t. $T$ and $\tilde h: \RE \to \CO$ by prolonging $h_s$ to zero. By Lemma \ref{l:prolong}-$ii)$, we have that $\tilde h \in  H^{3/4}(\RE)$. Hence, by Lemma \ref{l:reg}, $I^{1/2}\tilde h\in \dot H^{5/4}(\RE)$ and $\frac{d}{dt}I^{1/2}\tilde h$ is in $ \dot H^{1/4}(\RE)$.\\ 
Since $\tilde h$ coincides with $h$ for all $t\in[0,T]$ in the integral in Eq. \eqref{5.43a} we can write $f=\tilde h+  f_0  $, so that 
\begin{align*}
\left|\int_0^t \dd  s I^{1/2}(q^\ve -q)(s) \frac{d}{dt}I^{1/2} f (t-s) \right| \leq & C \|I^{1/2}(q^\ve   - q) \|_{L^\infty(0,T)}\left( \left\|\frac{d}{dt}I^{1/2} \tilde h\right\|_{L^1(0,T)}  + |f_0| \sqrt T\right) \\ \leq &  C \|I^{1/2}(q^\ve   - q) \|_{L^\infty(0,T)}
\end{align*}
for all $t\in[0,T]$. Here  we used 
 \[
 \frac{d}{dt} I^{1/2} f_0 (t) =  \frac{f_0}{\sqrt t},
  \]
  and the fact that $\frac{d}{dt}I^{1/2}\tilde h$ is in $L^2_{loc}(\RE)$. We have proved that 
\[
A^\ve(t) \leq C \|I^{1/2}(q^\ve   - q) \|_{L^\infty(0,T)},\] for all $t\in[0,T]$. 
To conclude the proof of inequality \eqref{black2} it remains to show that 
\begin{equation} \label{smith} 
R^\ve(t)\leq C\ve^{1/2}.
\end{equation}  We start by noting that  
\[
R^\ve(t) \leq \|q^\ve\|_{L^\infty(0,T)}    \int_0^t \dd s \, \left| (\psi(t),  \big (U(t-s) \rho^\ve  - U(t-s,\cdot )) \big) \right| .  \]
Then we show that 
\begin{equation}\label{5.46}
\left| (\psi(t),  \big (U(t) \rho^\ve  - U(t,\cdot )) \big)\right|  \leq  C\left(1 + \frac{1}{t^{3/4}}\right)\ve^{1/2}
  \end{equation}
which in turn implies 
\[
 \int_0^t \dd s \, \left| (\psi(t),  \big (U(t-s) \rho^\ve  - U(t-s,\cdot )) \big)\right|  \leq  C\ve^{1/2}
 \]
and consequently bound \eqref{smith} follows. To prove the claim in \eqref{5.46} we use again the decomposition $\psi(t) = \phi(t) + q(t) G$, so that 
  \[
  \left| (\psi(t),  \big (U(t) \rho^\ve  - U(t,\cdot )) \big) \right|   \leq  \widetilde R_{1}^\ve(t) + \widetilde R_{2}^\ve(t).
\] 
with 
\[
\widetilde R_{1}^\ve(t) =  \, \left| (\phi(t),  \big (U(t) \rho^\ve  - U(t,\cdot )) \big) \right| 
\]
and 
\[
\widetilde R_{2}^\ve(t) = |q(t)| \left| (G,  \big (U(t) \rho^\ve  - U(t,\cdot )) \big) \right| .
\]
We use the Fourier transform and find  
\[
\widetilde R_{1}^\ve(t) \leq   \left(  \int_{|\kk|\leq\ve^{-1}} \dk  |\hat\phi(t)(\kk)| \,  |\hat \rho(\ve k)  - \hat \rho(0)|  +  \int_{|\kk|>\ve^{-1}} \dk  |\hat\phi(t)(\kk)| \,  |\hat \rho(\ve k)  -\hat \rho(0)|  \right).
\]
Since $ |\hat \rho( k)  -\hat \rho(0)| \leq C k^\eta$, for all $0\leq\eta\leq 2$, then  
\[
 \int_{|\kk|\leq\ve^{-1}} \dk  |\hat\phi(t)(\kk)| \,  |\hat \rho(\ve k)  -\hat \rho(0)| \leq \ve^2  C  \int_{|\kk|\leq\ve^{-1}} \dk  |\hat\phi(t)(\kk)|  |\kk|^2 \leq C \|\Delta \phi(t)\| \ve^{1/2} 
\]
and    
\[
 \int_{|\kk|>\ve^{-1}} \dk  |\hat\phi(t)(\kk)| \,  |\hat \rho(\ve k)  -\hat \rho(0)|  \leq    \|\Delta\phi(t)\| \left(\int_{|\kk|>\ve^{-1}} \dk  \frac{ |\hat \rho(\ve k)  -\hat \rho(0)|^2}{|\kk|^4}\right)^{1/2}   \leq C \ve^{1/2},
\]
 hence $\widetilde R_{1}^\ve(t) \leq C \ve^{1/2}$. 

Next we analyze $\widetilde R_{2}^\ve(t)$. We have that   
\[ \widetilde R_{2}^\ve(t) = \frac{|q(t)|}{2\pi^2} \left|   \int_0^\infty \dd k \, e^{-ik^2(t-s)}   (\hat \rho(\ve k)  - \hat \rho(0))\right|. \] 
We use the bound 
\begin{equation}\label{5.47}
\left|\int_0^\infty \dd k\, ( \hat \rho(\ve k) -\hat\rho(0))\, e^{-ik^2t}   \right| \leq C \frac{\ve^\eta}{t^{(1+\eta)/2}} \qquad \forall 0\leq \eta < 1,
\end{equation}
that can be proved exactly in the same way as bound \eqref{5.38}, setting $\eta =1/2 $ we get 
$ \widetilde R_{2}^\ve(t) \leq C \ve^{1/2}/t^{3/4}$. And this concludes the proof of \eqref{5.46}. 

\n
It remains to  prove \eqref{black3}. We write 
\begin{multline*}
  \left| \int_0^t ds  |q^\ve(s)|^{2\mu}q^\ve(s) (\psi(t), U(t-s) \rho^\ve) \right|   \\
  \leq 
     \left| \int_0^t ds  (\psi(t),  |q^\ve(s)|^{2\mu}q^\ve(s)  U(t-s) \rho^\ve -   |q(s)|^{2\mu}q(s)  U(t-s,\cdot )) \right|  \\ 
      +
   \left| \int_0^t ds  |q(s)|^{2\mu}q(s)   (\psi(t), U(t-s,\cdot ))  \right|.
  \end{multline*}
By using the same argument used in the proof of \eqref{black2} it is possible to show that
\[
\begin{aligned}
  &  \left| \int_0^t ds  (\psi(t),  |q^\ve(s)|^{2\mu}q^\ve(s)  U(t-s) \rho^\ve -   |q(s)|^{2\mu}q(s)  U(t-s,\cdot )) \right|     \\ 
  \leq & C \left(\|I^{1/2}(|q^\ve|^{2\mu}q^\ve - |q|^{2\mu}q)\|_{L^\infty(0,T)} +\ve^{1/2}\right) \leq C.
\end{aligned}
\]
Moreover $(\psi(t), U(t,\cdot))   = (U(t)\bar \psi(t)) (\oo)  $ is integrable in $[0,T]$ because $ (U(\cdot)\bar \phi(t)) (\oo) \in \dot H^{3/4}(0,T)$ and $ (U(t)G) (\oo)$  is integrable by Eq. \eqref{ray}. Hence
\[
\left| \int_0^t ds  |q(s)|^{2\mu}q(s)   (\psi(t), U(t-s,\cdot ))  \right| \leq C. 
\]
Inequality \eqref{black3} follows by taking into account these two bounds and the factor $\ve$ in \eqref{black3}.
\end{proof}
To prove Th.\ref{t:main} we need to show that $\|I^{1/2}(q^\ve -q)\|_{L^\infty(0,T)} $ goes to zero with $\ve$, which is done in Lemma \ref{l:IDelta} below. Here we start by finding a representation  for  $I^{1/2}(q^\ve -q)$, easily obtained using  \eqref{limit2} and \eqref{Iqve}. We apply the operator $I^{1/2}$ to Eq. \eqref{limit2} and obtain 
\begin{equation}\label{yummy}
I^{1/2}q(t)+4\pi \sqrt{\pi i}\, \ga \int_0^t \ds |q(s)|^{2\mu} q(s) =  4 \pi\sqrt{\pi i} \int_0^t \ds(U(s)\psi_0)(\oo). 
\end{equation}
Taking the difference between Eqs. \eqref{yummy} and \eqref{Iqve} we obtain 
\begin{equation}\label{yummy0}
I^{1/2} (q^\ve - q)(t)  +4\pi \sqrt{\pi i}\, \ga \int_0^t \ds ( |q^\ve(s)|^{2\mu} q^\ve(s) - |q(s)|^{2\mu} q(s) ) ={Y}^\ve(t). 
\end{equation}
with 
\[
 Y^\ve(t) = \sum_{j=1}^4 Y_j^\ve(t)  ,
\]
where $Y_1^\ve(t)$,  $Y_2^\ve(t)$, and $Y_3^\ve(t)$  were given in  Eqs. \eqref{Yve1} -  \eqref{Yve3}, and $Y_4^\ve(t)$ is defined as 
\begin{equation}\label{Yve4}Y_4^\ve(t)   = 4\pi \sqrt{\pi i}\left( \int_0^t \dd s \, (\rho^\ve,U(s)\psi^\ve_0) - \int_0^t \ds(U(s)\psi_0)(\oo) \right) . \end{equation}

To prove Lemma \ref{l:IDelta}  we shall need the following bounds on the remainder $Y^\ve.$
\begin{lemma}\label{l:Yve}Let $\psi_0=\phi_0+q_0G \in \DD$ and $\psi_0^\ve=\phi_0+q_0\rho^\ve*G$ and fix $T>0$. Then if $\gamma\geq 0$ and $\mu\geq 0$ or $\gamma<0$ and $0\leq\mu<1$, there exist two positive constants  $\ve_0$ and $C$  such that, for all  $0<\ve<\ve_0$, 
\begin{equation}\label{bounderrore}
\|{Y}^{\ve}\|_{L^{\infty}(0,T)} \leq C \ve^{1/2}  
\end{equation}
and 
\begin{equation}\label{boundD1/2errore}
\| D^{1/2}{Y}^{\ve}(t)\|_{L^1(0,T)} \leq C\ve^{\delta}  
\end{equation}
for all $0\leq \de <1/2$ with the constant $C$ in \eqref{boundD1/2errore} depending on $\de$.
\end{lemma}
\begin{proof}
We start with the proof of   bound \eqref{bounderrore}. We  use  
\[
\|Y^\ve\|_{L^{\infty}(0,T)}\leq   \sum_{j=1}^4 \| Y_j^\ve\|_{L^{\infty}(0,T)}\] and bound the sum term by term. 
For  $ \| Y_1^\ve\|_{L^{\infty}(0,T)}$ we have that 
\[
\|Y_1^\ve\|_{L^{\infty}(0,T)} \leq C \int_0^T \dd \tau\left|     \int_0^\infty \dd k\, ( (\hat \rho(\ve k))^2 -(\hat\rho(0))^{2})\, e^{-ik^2(t-\tau)} \right| \leq C \ve^\eta 
\]
for all $0\leq \eta <1$, where we used  bounds \eqref{apriori1} and \eqref{5.38}. With a similar argument we find  
\[
\|Y_2^\ve\|_{L^{\infty}(0,T)} \leq C \ve  \int_0^T \dd \tau    \left|  \int_0^\infty \dd k\, ( (\hat \rho(\ve k))^2 -(\hat\rho(0))^{2})\, \, e^{-ik^2(t-\tau)} \right| \leq C \ve^{\eta+1}
\]for all $0\leq\eta<1$. 
Moreover, using again bound \eqref{apriori1}, we have that 
\[\|Y_3^\ve\|_{L^{\infty}(0,T)} \leq C\ve  .\]
Next we find a bound for $Y_4^\ve$. We start by rewriting 
\begin{equation}\label{5.51}
Y_4^\ve(t) =4\pi \sqrt{\pi i}\left(  \int_0^t \dd s \, (\rho^\ve,U(s)(\psi^\ve_0 - \psi_0))   + \int_0^t \dd s \, \left((\rho^\ve,U(s)\psi_0) - (U(s)\psi_0)(\oo) \right) \right).  
\end{equation}
Hence
\[\|Y_4^\ve\|_{L^{\infty}(0,T)} \leq C \left(  \int_0^T \dd s \, |(\rho^\ve,U(s)(\psi^\ve_0 - \psi_0))|   + \int_0^T \dd s \, |(\rho^\ve,U(s)\psi_0) - (U(s)\psi_0)(\oo)| \right) .  
\]Recalling that $\rho^\ve$  is real one has that the second term in the latter identity can be written as  
\[ (\rho^\ve,U(s)\psi_0) - (U(s)\psi_0)(\oo)  = \left(\bar\psi_0, U(s) \rho^\ve- U(s,\cdot ) \right), \] hence by  \eqref{5.46}  it follows that  
\[ \int_0^T \dd s \, \left|(\rho^\ve,U(s) \psi_0) - (U(s) \psi_0)(\oo) \right| \leq C\ve^{1/2} . \] 
To bound the first term in Eq. \eqref{5.51} we note the identity  
\[\begin{aligned} &   (\rho^\ve,U(t)(\psi^\ve_0 - \psi_0))  = q_0  (\rho^\ve,U(t) (\rho^\ve * G - G)) \\ 
= & 4\pi q_0  \int_0^\infty \dd k e^{-ik^2 t} \left((\hat\rho(\ve k))^2 - (\hat\rho(0))^2\right) 
- 4\pi q_0 \hat\rho(0)  \int_0^\infty \dd k e^{-ik^2 t} \left(\hat\rho(\ve k) - \hat\rho(0)\right) . \end{aligned}\]
Hence, by using  bounds \eqref{5.38}  and \eqref{5.47} one can prove that 
\begin{equation}
\label{fangs}
| (\rho^\ve,U(t) (\psi^\ve_0 - \psi_0))| \leq C \frac{\ve^\eta}{t^{(1+\eta)/2}} \qquad \forall\, 0\leq \eta < 1,
\end{equation}which implies 
\[  \int_0^T \dd s \, \left|(\rho^\ve,U(s)(\psi^\ve_0 - \psi_0))\right| \leq C \ve^\eta \] for all $0\leq\eta<1$. Hence \begin{equation}\label{boundYve4}
\|Y_4^\ve\|_{L^{\infty}(0,T)} \leq C \ve^{1/2}. \end{equation}\\
\noindent
Next  we prove bound \eqref{boundD1/2errore}. \\
Again we use $ \|D^{1/2} Y^\ve\|_{L^1(0,T)}\leq  \sum_{j=1}^4\|D^{1/2} Y_j^\ve\|_{L^1(0,T)}$, and  bound the sum  term by term.  \\
We start the analysis from the term $Y_1^\ve(t)$,  we shall prove that 
\begin{equation}\label{sorrowful}  \|D^{1/2} Y_1^\ve\|_{L^1(0,T)}\leq C\ve^{\de} \qquad \forall\,0\leq\de<\frac12.  \end{equation}  
To begin with we compute $I^{1/2}Y_1^\ve$.  We have that 
\[\begin{aligned}
I^{1/2}Y_1^\ve(t) = &C  \int_0^t \dd s \frac{1}{\sqrt{t-s}} \int_0^s \dd \tau q^\ve(\tau)    \int_0^\infty \dd k\, ( (\hat \rho(\ve k))^2 - (\hat \rho(0))^2)\, e^{-ik^2(s-\tau)} \\ 
= &C\int_0^t \dd \tau q^\ve(\tau)  \int_0^\infty \dd k  ( (\hat \rho(\ve k))^2 - (\hat \rho(0))^2)\, e^{ik^2\tau}   \int_\tau^t \dd s \frac{1}{\sqrt{t-s}}    e^{-ik^2s} \\
= &C \int_0^t \dd \tau q^\ve(\tau)  \int_0^\infty \dd k  ( (\hat \rho(\ve k))^2 - (\hat \rho(0))^2 ) \, \frac{e^{-ik^2(t-\tau)}}{k}   \int_0^{k\sqrt{t-\tau}} \dd \xi \,  e^{i\xi^2} 
\end{aligned}\]
where we exchanged order of integration and used the change of variables $s \to \xi = k \sqrt{t-s}$.  Next we use the  change of variables $k\to p = k\sqrt{t-\tau}$, the identity $e^{-ip^2} = -\frac{1}{2ip} \frac{d}{dp} e^{-ip^2}$ and  integrate by parts, changing also $\tau \to t-\tau $, we  obtain 
\[\begin{aligned}
I^{1/2}Y_1^\ve(t) 
= &C   \int_0^t \dd \tau q^\ve(t-\tau)  \int_0^\infty \dd p  ( (\hat \rho(\ve p/ \sqrt{\tau}))^2 - (\hat \rho(0))^2 ) \, \frac{ e^{-ip^2}}{p}   \int_0^{p} \dd \xi \,  e^{i\xi^2}  \\ 
= & C    \int_0^t \dd \tau q^\ve(t-\tau)  \int_0^\infty \dd p  \frac{ (\hat \rho(\ve p/ \sqrt{\tau}))^2 - (\hat \rho(0))^2 }{p^2} \, \frac{d}{dp} e^{-ip^2}   \int_0^{p} \dd \xi \,  e^{i\xi^2}  \\ 
=& \sum_{j=1}^3 B_j^\ve(t)\\ 
\end{aligned}\]
where 
\[\begin{aligned} 
B_1^\ve(t) &=C \int_0^t \dd \tau q^\ve(t-\tau)  \int_0^\infty \dd p  \frac{ (\hat \rho(\ve p/ \sqrt{\tau}))^2 - (\hat \rho(0))^2 }{p^2} \\
B_2^\ve(t) &=C  \int_0^t \dd \tau q^\ve(t-\tau)  \int_0^\infty \dd p  \frac{ (\hat \rho(\ve p/ \sqrt{\tau}))^2 - (\hat \rho(0))^2 }{p^3} \,  e^{-ip^2}   \int_0^{p} \dd \xi \,  e^{i\xi^2}  \\
B_3^\ve(t) &= C \ve \int_0^t \dd \tau \frac{q^\ve(t-\tau)}{\sqrt{\tau}}  \int_0^\infty \dd p  \frac{ \hat \rho'(\ve p/ \sqrt{\tau}) \hat \rho(\ve p/ \sqrt{\tau}) }{p} \, e^{-ip^2}   \int_0^{p} \dd \xi \,  e^{i\xi^2},
\end{aligned}\]
and where we used the fact that 
\[\lim_{p\to0}  \frac{ (\hat \rho(\ve p/ \sqrt{\tau}))^2 - (\hat \rho(0))^2 }{p^2} \, e^{-ip^2}   \int_0^{p} \dd \xi \,  e^{i\xi^2}  =0  .\]
Taking the derivative we find 
\[D^{1/2}Y_1^\ve(t) = \sum_{j=1}^3\dot B_j^\ve(t), \]
hence
\[\| D^{1/2}  {Y_1}^{\ve}\|_{L^1(0,T)} \leq \sum_{j=1}^3 \| \dot B_j^\ve\|_{L^1(0,T)}.\]
 We bound the r.h.s. term by term. \\
Term $B_1^\ve(t)$. By scaling the integral over $p$ we rewrite it as 
 \begin{equation}\label{dotB1ve} B_1^\ve(t) = C \ve \int_0^t \dd \tau \frac{q^\ve(t-\tau)}{\sqrt\tau} \ . 
  \end{equation}
Hence $\dot B_1^\ve(t)=  C  \ve   D^{1/2}q^\ve(t)$ and,   by Eq. \eqref{apriori3}, $ \| \dot B_1^\ve\|_{L^1(0,T)}\leq C\ve^{\de}$, for all $0\leq\de<1/2$.  \\
Term $B_2^\ve(t)$. Taking the derivative we get 
  \[\begin{aligned}
  \dot B_2^\ve(t) = & C\bigg(  \int_0^t \dd \tau \dot q^\ve(t-\tau)  \int_0^\infty \dd p  \frac{ (\hat \rho(\ve p/ \sqrt{\tau}))^2 - (\hat \rho(0))^2 }{p^3} \,  e^{-ip^2}   \int_0^{p} \dd \xi \,  e^{i\xi^2}  \\ 
  & + q^\ve(0)  \int_0^\infty \dd p  \frac{ (\hat \rho(\ve p/ \sqrt{t}))^2 - (\hat \rho(0))^2 }{p^3} \,  e^{-ip^2}   \int_0^{p} \dd \xi \,  e^{i\xi^2} \bigg). 
  \end{aligned}
  \]
  We use bounds \eqref{working},   \eqref{apriori2} and the fact that $\left| \int_0^{p} \dd \xi \,  e^{i\xi^2}\right| \leq C \min\{1,p\}$ and obtain
\begin{multline*}
\left|\int_0^t \dd \tau \dot q^\ve(t-\tau)  \int_0^\infty \dd p  \frac{ (\hat \rho(\ve p/ \sqrt{\tau}))^2 - (\hat \rho(0))^2 }{p^3} \,  e^{-ip^2}   \int_0^{p} \dd \xi \,  e^{i\xi^2}\right| \\ \leq  
 C \ve^{-3/2 } \int_0^t \dd \tau \left( \int_0^1 \dd p  \left(\frac{\ve p}{\sqrt{\tau}}\right)^\eta \frac1{p^2} + \int_1^\infty \dd p  \left(\frac{\ve p}{\sqrt{\tau}}\right)^\eta \frac{1}{p^3} \right) \\ 
\leq   C \ve^{\de } \int_0^t \dd \tau \frac{1}{\tau^{\frac34+\frac\de2}} \left( \int_0^1 \dd p \frac1{p^{\frac12-\de}} + \int_1^\infty \dd p   \frac{1}{p^{\frac32-\de}} \right)  \leq C\ve^\de
  \end{multline*}
  where we set $\eta = \frac32 +\de$ with $0\leq\de <1/2$. Similarly, using  bound   \eqref{apriori1} instead of  \eqref{apriori3}, we get  
\[  \left|q^\ve(0)\int_0^\infty \dd p  \frac{ (\hat \rho(\ve p/ \sqrt{t}))^2 - (\hat \rho(0))^2 }{p^3} \,  e^{-ip^2}   \int_0^{p} \dd \xi \,  e^{i\xi^2} \right|   \leq C\frac{\ve^{\frac32+\de}}{t^{\frac34+\frac\de2}}
  \]
hence, $ \| \dot B_2^\ve\|_{L^1(0,T)}\leq C\ve^{\de}$, for all $0\leq\de<1/2$.\\
The bound on the term $B_3^\ve(t)$ requires a more careful analysis, this is due to the fact that the measure $\dd p/p$ is invariant under scaling on $p$. As a first step we split the integral over $p$ for small and large values of $p$ and set 
\[
B_3^\ve(t)
= B_{31}^\ve(t) + B_{32}^\ve(t)
\]
with 
\[B_{31}^\ve(t) = C \ve \int_0^t \dd \tau \frac{q^\ve(t-\tau)}{\sqrt{\tau}}  \int_0^1 \dd p  \frac{ \hat \rho'(\ve p/ \sqrt{\tau}) \hat \rho(\ve p/ \sqrt{\tau}) }{p} \, e^{-ip^2}   \int_0^{p} \dd \xi \,  e^{i\xi^2},\]and
\[B_{32}^\ve(t) =  C \ve \int_0^t \dd \tau \frac{q^\ve(t-\tau)}{\sqrt{\tau}}  \int_1^\infty \dd p  \frac{ \hat \rho'(\ve p/ \sqrt{\tau}) \hat \rho(\ve p/ \sqrt{\tau}) }{p} \, e^{-ip^2}   \int_0^{p} \dd \xi \,  e^{i\xi^2}.\]
The term $B_{31}^\ve(t)$ can be bounded without performing additional manipulations. We take the derivative and obtain 
\[\begin{aligned}
 \dot B_{31}^\ve(t) =  & C\bigg( \ve \int_0^t \dd \tau \frac{\dot q^\ve(t-\tau)}{\sqrt{\tau}}  \int_0^1 \dd p  \frac{ \hat \rho'(\ve p/ \sqrt{\tau}) \hat \rho(\ve p/ \sqrt{\tau}) }{p} \, e^{-ip^2}   \int_0^{p} \dd \xi \,  e^{i\xi^2} \\ 
& + \ve \frac{q^\ve(0)}{\sqrt{t}}  \int_0^1 \dd p  \frac{ \hat \rho'(\ve p/ \sqrt{t}) \hat \rho(\ve p/ \sqrt{t}) }{p} \, e^{-ip^2}   \int_0^{p} \dd \xi \,  e^{i\xi^2}\bigg)
  \end{aligned}
  \]
  We use  bounds \eqref{working2},   \eqref{apriori2} and the fact that $\left| \int_0^{p} \dd \xi \,  e^{i\xi^2}\right| \leq Cp$ and obtain
\begin{multline*}
\left|  \ve \int_0^t \dd \tau \frac{\dot q^\ve(t-\tau)}{\sqrt{\tau}}  \int_0^1 \dd p  \frac{ \hat \rho'(\ve p/ \sqrt{\tau}) \hat \rho(\ve p/ \sqrt{\tau}) }{p} \, e^{-ip^2}   \int_0^{p} \dd \xi \,  e^{i\xi^2} \right| \\ 
\leq C  \ve^{-1/2} \int_0^t \dd \tau \frac{1}{\sqrt{\tau}} 
\int_0^1 \dd p  \left( \frac{\ve p}{ \sqrt{\tau}}\right)^\eta \leq C  \ve^{\de} \int_0^t \dd \tau \frac{1}{\tau^{\frac34+\frac\de2}}  
\leq C\ve^\de
\end{multline*}
where we set $\eta = \frac12+\de$ with $0\leq \de<1/2$. Similarly   
\[
\left|\ve \frac{q^\ve(0)}{\sqrt{t}}  \int_0^1 \dd p  \frac{ \hat \rho'(\ve p/ \sqrt{t}) \hat \rho(\ve p/ \sqrt{t}) }{p} \, e^{-ip^2}   \int_0^{p} \dd \xi \,  e^{i\xi^2}
\right|\leq  C  \frac{\ve^{\frac32+\de}}{t^{\frac34+\frac\de2}}  .\]
Hence  $ \| \dot B_{31}^\ve\|_{L^1(0,T)}\leq C\ve^{\de}$, for all $0\leq\de<1/2$. \\ 
To bound the term $B_{32}^\ve(t)$ we use again the identity $e^{-ip^2} = -\frac{1}{2ip} \frac{d}{dp} e^{-ip^2}$  and integrate
 by parts, we get  
 \[\begin{aligned}
B_{32}^\ve(t) =&   C\ve \int_0^t \dd \tau \frac{q^\ve(t-\tau)}{\sqrt{\tau}}  \int_1^\infty \dd p  \frac{ \hat \rho'(\ve p/ \sqrt{\tau}) \hat \rho(\ve p/ \sqrt{\tau}) }{p^2} \,   \frac{d}{dp} e^{-ip^2}   \int_0^{p} \dd \xi \,  e^{i\xi^2} \\ 
= &\BB_1^\ve(t)+ \BB_2^\ve(t) + \BB_3^\ve(t)+ \BB_4^\ve(t) 
 \end{aligned}\]
with 
\[\begin{aligned}
 \BB_1^\ve(t) =&C \ve \int_0^t \dd \tau \frac{q^\ve(t-\tau)}{\sqrt{\tau}}   \hat \rho'(\ve / \sqrt{\tau}) \hat \rho(\ve / \sqrt{\tau}) e^{-i}   \int_0^{1} \dd \xi \,  e^{i\xi^2},\\
\BB_2^\ve(t) = &C \ve \int_0^t \dd \tau \frac{q^\ve(t-\tau)}{\sqrt{\tau}}  \int_1^\infty \dd p  \frac{ \hat \rho'(\ve p/ \sqrt{\tau}) \hat \rho(\ve p/ \sqrt{\tau}) }{p^2} ,\\
\BB_3^\ve(t) = &C \ve \int_0^t \dd \tau \frac{q^\ve(t-\tau)}{\sqrt{\tau}}  \int_1^\infty \dd p  \frac{ \hat \rho'(\ve p/ \sqrt{\tau}) \hat \rho(\ve p/ \sqrt{\tau}) }{p^3}  e^{-ip^2}  \int_0^{p} \dd \xi \,  e^{i\xi^2} \\
\BB_4^\ve(t) =& C \ve^2 \int_0^t \dd \tau \frac{q^\ve(t-\tau)}{\tau}  \int_1^\infty \dd p  \frac{ \hat \rho''(\ve p/ \sqrt{\tau}) \hat \rho(\ve p/ \sqrt{\tau})  +( \hat \rho'(\ve p/ \sqrt{\tau}))^2}{p^2} \,  e^{-ip^2}   \int_0^{p} \dd \xi \,  e^{i\xi^2}
\end{aligned}
\]
For the terms $\BB_1^\ve(t)$, $\BB_2^\ve(t)$, and $\BB_3^\ve(t)$, one can take the derivative and use  bounds  \eqref{working2}, \eqref{apriori1}, \eqref{apriori2} and the bound $\left| \int_0^{p} \dd \xi \,  e^{i\xi^2}\right| \leq C $. We write down the details of the analysis for the term $\BB_2^\ve(t)$ only, the other two are similar. Taking the derivative one obtains 
\[\dot {\BB}_2^\ve(t) =    C\bigg(  \ve \int_0^t \dd \tau \frac{\dot q^\ve(t-\tau)}{\sqrt{\tau}}  \int_1^\infty \dd p  \frac{ \hat \rho'(\ve p/ \sqrt{\tau}) \hat \rho(\ve p/ \sqrt{\tau}) }{p^2} +  \ve \frac{q^\ve(0)}{\sqrt{t}}  \int_1^\infty \dd p  \frac{ \hat \rho'(\ve p/ \sqrt{t}) \hat \rho(\ve p/ \sqrt{t }) }{p^2}\bigg).  
\]
We have that 
\begin{multline*}
\left| 
\ve \int_0^t \dd \tau \frac{\dot q^\ve(t-\tau)}{\sqrt{\tau}}  \int_1^\infty \dd p  \frac{ \hat \rho'(\ve p/ \sqrt{\tau}) \hat \rho(\ve p/ \sqrt{\tau}) }{p^2}  \right| \leq 
C\ve^{-\frac12} \int_0^t \dd \tau \frac{1}{\sqrt{\tau}}  \int_1^\infty \dd p  \left(\frac{\ve p}{ \sqrt{\tau}}\right)^\eta \frac{1}{p^2}   \\ 
\leq 
C\ve^{\de} \int_0^t \dd \tau \frac{1}{\tau^{\frac34+\frac\de2}}  \int_1^\infty \dd p  \frac{1}{p^{\frac32-\de}}   \leq C\ve^\de
\end{multline*}
where we set $\eta = \frac12+\de$ with $0\leq\de<1/2$. Similarly  
\[
\left| \frac{q^\ve(0)}{\sqrt{t}}  \int_1^\infty \dd p  \frac{ \hat \rho'(\ve p/ \sqrt{t}) \hat \rho(\ve p/ \sqrt{t }) }{p^2} \right| \leq C  \frac{\ve^{\frac32+\de}}{t^{\frac34+\frac\de2}}.
\]
Hence $ \| \dot \BB_2^\ve\|_{L^1(0,T)}\leq C\ve^{\de}$, for all $0\leq\de<1/2$. Similarly,  $ \| \dot \BB_1^\ve\|_{L^1(0,T)}\leq C\ve^{\de}$ and $ \| \dot \BB_3^\ve\|_{L^1(0,T)}\leq C\ve^{\de}$.\\ 
Next we find a bound for  the term $\BB_4^\ve(t)$. Taking the derivative gives
\[
\begin{aligned}
\dot\BB_4^\ve(t) = & C\bigg(  \ve^2 \int_0^t \dd \tau \frac{\dot q^\ve(t-\tau)}{\tau}  \int_1^\infty \dd p  \frac{ \hat \rho''(\ve p/ \sqrt{\tau}) \hat \rho(\ve p/ \sqrt{\tau})  +( \hat \rho'(\ve p/ \sqrt{\tau}))^2}{p^2} \,  e^{-ip^2}   \int_0^{p} \dd \xi \,  e^{i\xi^2} \\ 
& + \ve^2 \frac{q^\ve(0)}{t}  \int_1^\infty \dd p  \frac{ \hat \rho''(\ve p/ \sqrt{t}) \hat \rho(\ve p/ \sqrt{t})  +( \hat \rho'(\ve p/ \sqrt{\tau}))^2}{p^2} \,  e^{-ip^2}   \int_0^{p} \dd \xi \,  e^{i\xi^2}\bigg).
\end{aligned}
\]
  We use  bound  \eqref{apriori2},  and the fact that $\left| \int_0^{p} \dd \xi \,  e^{i\xi^2}\right| \leq C$ and that $\frac1{p^2} \leq \frac{1}{p^{\frac12+\de}}$ for all $p>1$ and $0\leq\de<1/2$, hence 
 \begin{multline*}
\left| \ve^2 \int_0^t \dd \tau \frac{\dot q^\ve(t-\tau)}{\tau}  \int_1^\infty \dd p  \frac{ \hat \rho''(\ve p/ \sqrt{\tau}) \hat \rho(\ve p/ \sqrt{\tau})  +( \hat \rho'(\ve p/ \sqrt{\tau}))^2}{p^2} \,  e^{-ip^2}   \int_0^{p} \dd \xi \,  e^{i\xi^2} \right| \\ 
\leq C  \ve^{1/2} \int_0^t \dd \tau \frac{1}{\tau}  \int_1^\infty \dd p  \frac{ |\hat \rho''(\ve p/ \sqrt{\tau})| \, |\hat \rho(\ve p/ \sqrt{\tau})|  +( \hat \rho'(\ve p/ \sqrt{\tau}))^2}{p^{\frac12+\de}} \\ = 
C  \ve^{1/2} \int_0^t \dd \tau \frac{1}{\tau}  \left(\frac\ve{\sqrt\tau}\right)^{-\frac12+\de}\int_{\frac\ve{\sqrt\tau}}^\infty \dd p  \frac{ |\hat \rho''( p)| \, |\hat \rho(p)|  +( \hat \rho'( p))^2}{p^{\frac12+\de}} \\
\leq C  \ve^{\de} \int_0^t \dd \tau \frac{1}{\tau^{\frac34+\frac\de2}}\int_{0}^\infty \dd p  \frac{ |\hat \rho''( p)| \, |\hat \rho(p)|  +( \hat \rho'( p))^2}{p^{\frac12+\de}}\leq C  \ve^{\de}
 \end{multline*}
where we used the fact that   $(|\hat \rho''( p)| \, |\hat \rho(p)|  +( \hat \rho'( p))^2)p^{-(\frac12+\de)}$ is integrable.  With a similar argument  one can prove that 
\[
 \left| \ve^2 \frac{q^\ve(0)}{t}  \int_1^\infty \dd p  \frac{ \hat \rho''(\ve p/ \sqrt{t}) \hat \rho(\ve p/ \sqrt{t})  +( \hat \rho'(\ve p/ \sqrt{\tau}))^2}{p^2} \,  e^{-ip^2}   \int_0^{p} \dd \xi \,  e^{i\xi^2} \right|\leq C \frac{\ve^{\frac32+\de}}{t^{\frac34+\frac\de2}}.\]
 Hence $ \| \dot \BB_4^\ve\|_{L^1(0,T)}\leq C\ve^{\de}$ for all $0\leq\de<1/2$. And this concludes the proof of  bound \eqref{sorrowful}.  
 
The term $Y_2^\ve(t)$ is similar to the term $Y_1^\ve(t)$. One can perform exactly the same calculations just changing everywhere $q^\ve$ into $|q^\ve|^{2\mu}q^\ve$ and taking into account the  extra factor $\ve$.  All the terms obtained in this way, besides the one analogous to  $\dot B_1^\ve$, can be bounded exactly in the same way by taking into account bound \eqref{andatura}.  These terms indeed turn out to be smaller than the corresponding ones arising in the analysis of $Y_1^\ve$, due to the additional factor $\ve$. \\
For the term which is the analogous of $\dot B_1^\ve$ instead, one can use the definition of the operator $D^{1/2} $. Taking into account the extra $\ve$ factor, and Eq. \eqref{dotB1ve} it is easy to guess that the term analogous to $\dot B_1^\ve$ will be proportional to  
\[
\ve^2 \frac{d}{dt} \int_0^t  \dd \tau \frac{|q^\ve(t-\tau)|^{2\mu}q^\ve(t-\tau)}{\sqrt \tau}  = 
\ve^2 \int_0^t  \dd \tau   \frac{1}{\sqrt \tau}   \frac{d}{dt}|q^\ve(t-\tau)|^{2\mu}q^\ve(t-\tau) + \ve^2  \frac{|q^\ve(0)|^{2\mu}q^\ve(0)}{\sqrt t} .
\]
Hence by \eqref{apriori1} and $\eqref{andatura}$, 
\[
\left|\ve^2 \frac{d}{dt} \int_0^t  \dd \tau \frac{|q^\ve(t-\tau)|^{2\mu}q^\ve(t-\tau)}{\sqrt \tau} \right|   \leq C \left(
\ve^{\frac12} + \frac{\ve^2}{\sqrt t} \right).
\]
We conclude that the bound for $Y_2^\ve$ is 
\[
  \| D^{1/2} Y_2^\ve\|_{L^1(0,T)}\leq C\ve^{\frac12}. 
\]
The bound for the term $Y_3^\ve$ is easy. By an explicit calculation 
\[
D^{1/2}Y_3^\ve(t) =  \gamma\pi\frac{\ve}{\ell}  |q^\ve(t)|^{2\mu} q^\ve(t).
\]
Hence $|D^{1/2}Y_3^\ve(t) | \leq C \ve $ which in turn implies
\[
  \| D^{1/2} Y_3^\ve\|_{L^1(0,T)}\leq C\ve.
  \]
We are left to prove a bound on $Y_4^\ve$.  By a direct computation, see Eq. \eqref{Yve4}, one has that 
\[
D^{1/2}Y_4^\ve(t)   = 4\pi \sqrt{\pi i} \int_0^t \dd s  \frac{(\rho^\ve,U(s) \psi^\ve_0) - (U(s)\psi_0)(\oo)}{\sqrt{t-s}}. \] Reasoning like in the proof of bound \eqref{boundYve4} one can rewrite $D^{1/2}Y_4^\ve$ in a way similar to \eqref{5.51}, and obtain 
\[
D^{1/2}Y_4^\ve(t) = 4\pi \sqrt{\pi i}\left(   \int_0^t \dd s \, \frac{(\rho^\ve,U(s)(\psi^\ve_0 - \psi_0))}{\sqrt{t-s}}   + \int_0^t \dd s \, \frac{\left((\rho^\ve,U(s)\psi_0) - (U(s)\psi_0)(\oo)\right)}{\sqrt{t-s}}\right) .  
\]
Finally, by using bounds \eqref{fangs}, and \eqref{5.46} with $\psi_0$ instead of $\psi(t)$, it follows that 
\[|D^{1/2}Y_4^\ve(t)| \leq C\ve^{1/2}.\]
And this concludes the proof of bound \eqref{boundD1/2errore}. 
\end{proof}
In the following and last Lemma we control the half integral of the charge.
\begin{lemma}   \label{l:IDelta}
Let $\psi_0=\phi_0+q_0G \in \DD$ and $\psi_0^\ve=\phi_0+q_0\rho^\ve*G$ and fix $T>0$. Then if $\gamma\geq 0$ and $\mu\geq 0$ or $\gamma<0$ and $0\leq\mu<1$, and for all $0\leq\de<1/2$, there exist two positive constants  $\ve_0$ and $C$  such that, for all  $0<\ve<\ve_0$, 
\[
\| I^{1/2} (q^\ve -q) \|_{L^\infty(0,T)}  \leq  C \ve^{\de}. 
\]
\end{lemma}
\begin{proof}
Set $\Theta^\ve(t) =  (q^\ve - q)(t) $ and 
\[
g^{\ve}(s) =- 4\sqrt{\pi i} \frac{ |q^\ve(s)|^{2\mu} q^\ve(s) - |q(s)|^{2\mu} q(s) }{q^\ve(s) - q(s)}. 
\]
Since $||a|^{2\mu} a - |b|^{2\mu}b| \leq C(|a|^{2\mu}  + |b|^{2\mu} ) |a-b| $, for any $a,b\in\CO$, we have that $\|g^\ve\|_{L^\infty(0,T)} \leq C$. By Eq. \eqref{yummy0}  we have that 
\[
I^{1/2} \Theta^\ve (t) =  \pi \int_0^t \ds g^\ve(s)\Theta^\ve(s)+ {Y}^\ve(t). 
\]
Applying $D^{1/2}$ we obtain the identity 
\begin{equation}\label{eqDeve}
 \Theta^\ve(t) =  (I^{1/2}g^\ve\Theta^\ve)(t)+ \pi^{-1}D^{1/2} {Y}^\ve(t). 
\end{equation}
The solution of Eq. \eqref{eqDeve} is given by 
\[
\Theta^\ve(t) = \pi^{-1} \sum_{n=0}^\infty\left( (I^{1/2}g^\ve)^nD^{1/2}  Y^\ve\right)(t)  
= \pi^{-1}D^{1/2}  Y^\ve(t)   + \sum_{n=1}^\infty\left( (I^{1/2}g^\ve)^nD^{1/2}   Y^\ve\right)(t),
\]
hence 
\[
I^{1/2}\Theta^\ve(t) =   Y^\ve(t)   + \sum_{n=1}^\infty\left( I^{1/2}(I^{1/2}g^\ve)^nD^{1/2}   Y^\ve\right)(t).
\]
Concerning the series, one has that 
\begin{multline*}
\left\|\sum_{n=1}^\infty I^{1/2}(I^{1/2}g^\ve)^nD^{1/2}   Y^\ve\right \|_{L^\infty(0,T)} 
\\\leq 
\sum_{n=1}^\infty \left\|I^{1/2}(I^{1/2}g^\ve)^nD^{1/2}   Y^\ve\right \|_{L^\infty(0,T)}  
\leq  
\sum_{n=1}^\infty \left\|g^\ve\right \|_{L^\infty(0,T)}^n  
 \left\|I^{1/2}(I^{1/2})^n |D^{1/2}  Y^\ve| \right \|_{L^\infty(0,T)} 
 \\
 \leq  \pi \sum_{n=1}^\infty \left\|g^\ve\right \|_{L^\infty(0,T)}^n  
 \left\|(I^{1/2})^{n-1} 1 \right \|_{L^\infty(0,T)}  \left\| D^{1/2}   Y^\ve \right \|_{L^1(0,T)} 
 \leq  \ve^{\de}  \sum_{n=1}^\infty \frac{ C^n\left\|g^\ve\right \|_{L^\infty(0,T)}^n}{{\Gamma(\frac{n}2+1)}}   \leq C \ve^{\de}     
\end{multline*} 
with $0\leq \de<1/2$. Here $\Gamma$ denotes the gamma function, and the latter bound can be found in \cite[Eq. 7.2.6]{GF}, moreover we used  bound  \eqref{boundD1/2errore}. The statement then follows from the latter bound, from  bound \eqref{bounderrore} and from the trivial inequality 
\[
\|I^{1/2}\Theta^\ve\|_{L^\infty(0,T)} \leq \|  Y^\ve\|_{L^\infty(0,T)}   + \left\|\sum_{n=1}^\infty I^{1/2}(I^{1/2}g^\ve)^nD^{1/2}   Y^\ve\right\|_{L^\infty(0,T)}. 
\]
\end{proof}
We can now give the proof of Th. \ref{t:main}
\begin{proof}[Proof of Th. \ref{t:main}]
Bound \eqref{main} follows from Prop. \ref{l:psiconv} together with Prop. \ref{p:initdata} and Lemma \ref{l:IDelta}. 
\end{proof}

\end{document}